%% file: 1d_Performance_Evaluation_20.tex
\documentclass[a4paper,fleqn]{cas-sc}
\usepackage[numbers]{natbib}

\usepackage{booktabs} 

\usepackage{graphicx}
\usepackage{amsmath}
\usepackage{amsthm}
\usepackage[noend]{algpseudocode}
\usepackage{epstopdf}
\usepackage{algorithm}
\usepackage{color}
\usepackage{subcaption}
\usepackage{multirow}
\usepackage{framed}
\renewcommand{\P}{\mathbb{P}}

\newcommand{\np}[1]{\textcolor{black}{#1}}

\usepackage{gensymb}
\usepackage{latexsym,epsfig,color,verbatim}
\usepackage{pifont}
\usepackage{slashbox}
\usepackage{bbm}
\newtheorem{define}{Definition}
\newtheorem{remark}{\bf Remark}

\newtheorem{conjecture}{Conjecture}
\newtheorem{lemma}{Lemma}
\newtheorem{theorem}{Theorem}
\newtheorem{corollary}{Corollary}
\newtheorem*{theorem*}{\bf Theorem}

\begin{document}
\shorttitle{One-dimensional Distributed Service Networks}
\shortauthors{Panigrahy et~al.}
\title{Resource Allocation in One-dimensional Distributed Service Networks with Applications}

\author[1]{Nitish K. Panigrahy}
\cormark[1]
\ead{nitish@cs.umass.edu}

\author[2]{Prithwish Basu}
\ead{prithwish.basu@raytheon.com}

\author[3]{Philippe Nain}
\ead{philippe.nain@inria.fr}

\author[1]{Don Towsley}
\ead{towsley@cs.umass.edu}

\author[4]{Ananthram Swami}
\ead{ananthram.swami.civ@mail.mil}

\author[4]{Kevin S. Chan}
\ead{kevin.s.chan.civ@mail.mil}

\author[5]{Kin K. Leung}
\ead{kin.leung@imperial.ac.uk}

\address[1]{University of Massachusetts Amherst, MA 01003, USA}
\address[2]{Raytheon BBN Technologies, Cambridge, MA 02138, USA}
\address[3]{Inria, 06902 Sophia Antipolis Cedex, France}
\address[4]{Army Research Laboratory, Adelphi, MD 20783, USA.}
\address[5]{Imperial College London, London SW72AZ, UK.}

\cortext[cor1]{Corresponding author}
\cortext[cor1]{The material in this paper was presented in part at the IEEE International Symposium on Modeling, Analysis, and Simulation of Computer and Telecommunication Systems (MASCOTS), Rennes, France in 2019 and in Workshop on MAthematical performance Modeling and Analysis (MAMA 2018).}

%

\begin{abstract}
We consider assignment policies that allocate resources to users, where both resources and users are located on a one-dimensional line $[0,\infty)$. First, we consider unidirectional assignment policies that allocate resources only to users located to their left. We propose the Move to Right ({\it MTR}) policy, which scans from left to right assigning nearest rightmost available resource to a user, and contrast it to the Unidirectional Gale-Shapley ({\it UGS}) matching policy. While both policies among all unidirectional policies, minimize the expected distance traveled by a request ({\it request distance}), {\it MTR} is fairer. Moreover, we show that when user and resource locations are modeled by statistical point processes, and resources are allowed to satisfy more than one user, the spatial system under unidirectional policies can be mapped into bulk service queueing systems,  thus allowing the application of many queueing theory results that yield closed form expressions. As we consider a case where different resources can satisfy different numbers of users, we also generate new results for bulk service queues.  We also consider bidirectional policies where there are no directional restrictions on resource allocation and develop an algorithm for computing the optimal assignment which is more efficient than known algorithms in the literature when there are more resources than users. Numerical evaluation of performance of unidirectional and bidirectional allocation schemes yields design guidelines beneficial for resource placement. \np{Finally, we present a heuristic algorithm, which leverages the optimal dynamic programming scheme for one-dimensional inputs to obtain approximate solutions to the optimal assignment problem for the two-dimensional scenario and empirically yields request distances within a constant factor of the optimal solution.}
\end{abstract}


\maketitle

\section{Introduction}\label{sec:intro}
\input{01-introduction}

\section{Related Work}\label{sec:reltwrk}
\input{02-related-work}

\section{Technical Preliminaries}\label{sec:model}
\input{03-model}

\section{Unidirectional Allocation Policies}\label{sec:queue}
\input{04-queue}

\section {Unidirectional Poisson Matching}\label{sec:mm1}
\input{05-mm1}

\section{Unidirectional General Matching}\label{sec:mg1}
\input{06-mg1}
\section{Discussion of Unidirectional Allocation Policies}\label{sec:gen}
\input{07-generalizations}
\section{Bidirectional Allocation Policies}\label{sec:opt}
\input{08-opt-assignment}


\section{Numerical Experiments}\label{sec:perfcomp}
\input{09-perf-comp}

\np{
\section{Allocating Resources in 2D}\label{sec:2dto1d}
\input{10-2d-to-1d}
}
\section{Conclusion}\label{sec:con}
\input{11a-conclusion}

\section{Acknowledgment} 
\input{11b-ack}

\bibliographystyle{cas-model2-names}
\bibliography{refs} 

\section{Appendix}\label{appendix}
\input{11-appendix}
\end{document}

%% file: 01-introduction.tex

The past few years have witnessed significant growth in the use of distributed network analytics involving agile code, data and computational  resources. In many such networked systems, for example, Internet  of  Things \cite{atzori10}, a large number of computational and storage resources are  widely distributed in the physical world. These resources are accessed by various end users/applications that are also distributed over the physical space. Assigning users or applications to resources efficiently is key to the sustained high-performance operation of the system. 

In some systems, requests are transferred over a network to a server that provides a needed resource. In other systems, servers are mobile and physically move to the user making a request.  Examples of the former type of service include accessing storage resources over a wireless network to store files and requesting computational resources to run image processing tasks; whereas an example of the latter type of service is the arrival of ride-sharing vehicles to the user's location over a road transportation network.

Not surprisingly, the spatial distribution of resources and users\footnote{We use the terms ``users'' and ``requesters'' interchangeably and same holds true for the terms ``resources'' and ``servers''.} in the network is an important factor in determining the overall performance of the service. A key measure of performance is \emph{average request distance}, that is average distance between a user and its allocated resource/server (where distance is measured on the network). This directly translates to latency incurred by a user when accessing the service, which is arguably among the most important criteria in distributed service applications. For example, in wireless networks,  signal attenuation is strongly coupled to request distance, therefore developing allocation policies to minimize request distance can help reduce energy consumption, an important concern in battery-operated wireless networks. Another important practical constraint in distributed service networks is \emph{service capacity}. For example, in network analytics applications, a networked storage device can only support a finite number of concurrent users; similarly, a computational resource can only support a finite number of concurrent processing tasks. Likewise, in physical service applications like ride-sharing, a vehicle can pick up a finite number of passengers at once.

Therefore, a primary problem in such distributed service networks is to efficiently assign each user to a suitable resource so as to minimize average request distance and ensure no resource serves more users than its capacity. If the entire system is being managed by a single administrative entity such as a ride sharing service, or a datacenter network where analytics tasks are being assigned to available CPUs, there are economic benefits in minimizing the average request distance across all (user, resource) pairs, which is tantamount to minimizing the average delay in the system.

The general version of this capacitated assignment problem can be solved by modeling it as a \emph{minimum cost flow} problem on graphs~\cite{Ahuja93} and running the \emph{network simplex algorithm}~\cite{Orlin97}. However, if the network has a low-dimensional structure and some assumptions about the spatial distributions of users and resources hold, more efficient methods can be developed.


In this paper, we consider two one-dimensional network scenarios that motivate the study of this special case of the user-to-resource assignment problem. 

The first scenario is  ride-hailing on a one-way street where vehicles move right to left. If the vehicles of a ride-sharing company are distributed along the street at a certain time, and users equipped with smartphone ride-hailing apps request service, the system attempts to assign vehicles with spare capacity located towards the right of the users so as to minimize average ``pick up" distance. Abadi et al.~\cite{Abadi17} introduced this problem and presented a policy known as Unidirectional Gale-Shapley\footnote{We rename \emph{queue matching} defined in \cite{Abadi17}  as Unidirectional Gale-Shapley Matching to avoid overloading the term \emph{queue}.} matching ({\it UGS})  minimize average pick up distance. In this policy, all users concurrently emit rays of light toward their right and each user is matched with the vehicle that first receives the emitted ray. While the well-known Gale-Shapley matching algorithm~\cite{Gale62} matches user-resource pairs that are mutually nearest to each other, its unidirectional variant, UGS, matches a user to the nearest resource on its right. Note that, this one-dimensional network setting also applies to vehicular wireless ad-hoc networks on a one-lane roadway~\cite{Ho11,Leung94}\footnote{Furthermore, \cite{Ho11} confirms that vehicle location distribution on the streets in Central London can be closely approximated by a Poisson distribution.}, where users are in vehicles and servers are attached to fixed infrastructure such as lamp posts. Users attempt to allocate their computation tasks over the wireless network to servers located to their right so that they can retrieve the results with little effort while driving by.

In this paper, we propose another policy ``Move to Right'' policy (or {\it MTR}) which has the same ``expected distance traveled by a request'' ({\it request distance}) as UGS but has a lower variance. {\it MTR} sequentially allocates users to the geographically nearest available vehicle located to his/her right. When user and resource locations are modeled by statistical point processes the one-dimensional unidirectional space behaves similar to time and notions from queueing theory can be applied. In particular, when user and vehicle  locations are modeled by independent Poisson processes, average request distance can be characterized in closed form by considering inter-user and inter-server distances as parameters of a {\it bulk service} M/M/1 queue where the bulk service capacity denotes the maximum number of users that can be handled by a server. We equate request distance in the spatial system to the expected \emph{sojourn time} in the corresponding queuing model\footnote{Sojourn time is the sum of waiting and service times in a queue.}. This natural mapping allows us to use well-known results from queueing theory and in some cases to propose new queueing theoretic models to characterize request distances for a number of interesting situations beyond M/M/1 queues.

\np{A natural extension to our spatial framework is to consider more general communication costs associated with each resource allocation. Assuming communication cost for each allocation is a function of request distance, we provide closed form expressions for the expected communication cost for specific user-server distributions and specific server capacities.}

The second scenario involves a convoy of vehicles traveling on a one-dimensional space, for example, trucks on a highway or boats on a river. Some vehicles have expensive camera sensors (image/video) but have inadequate computational storage or processing power. On the other hand, cheap storage and processing is easily available on several other vehicles. The cameras periodically take photos/videos as they move through space and want them processed / stored. In such case, bidirectional assignment schemes are more suitable. Since no directionality restrictions are imposed on the allocation algorithms, computing the optimal assignment is not as simple as in the unidirectional case.

We explore the special structure of the one-dimensional topology to develop an optimal algorithm that assigns a set of requesters $R$ to a set of resources $S$ such that the total assignment cost is minimized. This problem has been recently solved for $|R|=|S|$~\cite{Bukac18}. However, we are interested in the case when $|R| < |S|$. We propose a dynamic Programming based algorithm which solves this case with time complexity $O(|R|(|S|-|R|+1))$. Note that other assignment algorithms in literature such as the Hungarian primal-dual algorithm and Agarwal's variant~\cite{Agarwal95} have time complexities $O(|R|^3)$ and $O(|R|^{2+\epsilon})$ respectively and assume $|R|=|S|$ for general and Euclidean distance measures.

\np{We leverage the optimal dynamic programming scheme for one-dimensional inputs to obtain approximate solutions to the optimal assignment problem for the two-dimensional scenario where users and servers are located on the two-dimensional plane, $\mathbb{R}^2$. More precisely, we \emph{embed} the points denoting $R,S \subset \mathbb{R^2}$ into new locations in $\mathbb{R}$ such that the distances between a user and its nearest servers are approximately preserved. Our approximation algorithm empirically yields request distances within a constant factor of the optimal solution with $O(|R|^2)$ time complexity.}

Our contributions are summarized below:
\begin{enumerate}
\item Analysis of simple unidirectional allocation policies {\it MTR} and {\it UGS} yielding closed form expressions for mean request distance. 
\begin{itemize}
\item When inter-requester and inter-resource distances are exponentially distributed, we model unidirectional policies as a bulk service M/M/1 queue.
\item When inter-requester distances are  generally distributed but the inter-resource distances are exponentially distributed, we model the situation using an accessible batch service G/M/1 queue.
\item When inter-requester distances are exponentially distributed but inter-resource distances are generally distributed, we model the spatial system as an accessible batch service M/G/1 queue with the first batch having exceptional service time. To the best of our knowledge this system has not been studied previously in the queueing theory literature.
\item We include several generalizations of our framework. In the first place we discuss a simulation driven conjecture for evaluating request distance for general distance distributions under heavy traffic. We also investigate the heterogeneous server capacity scenario where server capacity is a random variable and to the best of our knowledge this system has not been studied previously in the queueing theory literature. We derive expressions for expected request distance when servers have infinite capacity. We include communication cost associated with each resource allocation and provide a closed form expression for expected communication cost for some specific scenarios. Finally we extend our framework to compute expected request distance for the case where each user requests two resources residing in two different set of servers, by mapping it to a two queue fork-join system.
\end{itemize}
\item A novel algorithm for optimal (bidirectional) assignment with time complexity $O(|R|(|S|-|R|+1))$. 
\item A numerical and simulation study of different assignment policies: UGS , MTR, bi-directional heuristic allocation policies (Gale-Shapley \np{and Nearest Neighbor}) and the optimal policy.
\item \np{A heuristic based approximate solution to the optimal assignment problem for the two-dimensional scenario with an empirically observed constant factor approximation of the optimal solution.}
\end{enumerate}

The paper is organized as follows. The next section discusses related work. Section \ref{sec:model} contains technical preliminaries. We show the equivalence of UGS and MTR w.r.t expected request distance in Section~\ref{sec:queue}, and present results associated with the case when servers are Poisson distributed in Section ~\ref{sec:mm1}. In Section ~\ref{sec:mg1}, we develop formulations for expected request distance when either user or server placements are described by Poisson processes. We include some generalizations of our framework such as analysis under general distance distributions, results for heterogeneous server capacity and uncapacitated allocation in Section ~\ref{sec:gen}. The optimal bidirectional allocation strategy is presented in Section ~\ref{sec:opt}. We compare the  performance of various local allocation strategies in Section~\ref{sec:perfcomp}. \np{In Section ~\ref{sec:2dto1d}, we extend our one-dimensional framework to solve two-dimensional problem.} We conclude the paper in Section~\ref{sec:con}.

%% file: 02-related-work.tex

\noindent{\textit{Poisson Matching:}} Holroyd et al. \cite{Holroyd09} first studied translation invariant matchings between two $d$-dimensional Poisson processes with equal densities. Their primary focus was obtaining upper and lower bounds on expected matching distance for stable matchings. Abadi et al. \cite{Abadi17} introduced ``Unidirectional Gale-Shapley'' matching ({\it UGS}) and derived bounds on the expected matching distance for stable matchings between two one-dimensional Poisson processes with different densities. In this paper, we propose another unidirectional allocation policy: ``Move To Right'' policy ({\it MTR}) and provide explicit expressions for the expected matching distance for both MTR and UGS when either requesters or servers are distributed according to a renewal process and the according to a Poisson process.

\noindent{\textit{Exceptional Queueing Systems and Accessible Batches:}} Welch et al. \cite{Welch64} first studied an M/G/1 queue where a customer arriving when the server is idle has a different service time than the others. Bulk service M/G/1 queues has been studied in \cite{bailey54}. Authors in \cite{Goswami11} analyzed a bulk service G/M/1 queue with accessible or non-accessible batches where an accessible batch is considered to be a batch in service allowing subsequent arrivals, while the service is on. In this work, we model the spatial system using an accessible batch service queue with the first batch having exceptional service time. To the best of our knowledge this system has not been studied previously in queueing theory literature.

\noindent{\textit{Euclidean Bipartite Matching:}}
The optimal user-server assignment problem can be modeled as a minimum-weight matching on a weighted bipartite graph where weights on edges are given by the Euclidean distances between the corresponding vertices \cite{Mezard88}. Well-known polynomial time solutions exist for this problem, such as the modified Hungarian algorithm proposed by Agarwal et al. \cite{Agarwal95} with a running time of $O(|R|^{2+\epsilon})$, where $|R|$ is the total number of users. In the case of an equal number of users and servers, the optimal user-server assignment on a real line is known \cite{Bukac18}. In this paper, we consider the case when there are fewer users than servers.

%% file: 03-model.tex

Consider a set of users $R$ and a set of servers $S$. Each user makes a request that can be satisfied by any server. Assume that each server $j\in S$ has capacity $c_j \in \mathbb{Z}^+$ corresponding to the maximum number of requests that it can process. Suppose users and servers are located on a line $\mathcal{L}$. Formally, let $r : R \to \mathcal{L}$ and $s : S \to \mathcal{L}$ be the location functions for users and servers, respectively, such that a distance $d_{\mathcal{L}}(r,s)$ is well defined for all pairs $(r,s)\in R\times S$. Initially we assume that all servers have equal capacities i.e. $c_j = c\;\forall j \in S.$ Later in Section \ref{sec:het} we extend our analysis to a case in which server capacities are integer random variables.

\subsection{User and server spatial distributions}
Let $0\le r_1 \le r_2 \le \cdots$ represent user locations and $0\le s_1 \le s_2 \le \cdots$ be the server locations. Let  $X_j = s_j - s_{j-1}, j\ge 1, s_0 = 0,$ denote the inter-server distances and  $Y_i = r_i - r_{i-1}, i\ge 1, r_0 = 0,$ the inter-user distances. We assume $\{X_j\}_{j\ge1}$ to be a renewal process with cumulative distribution function (cdf) 
\begin{align}
\mathbb{P}(X_{j} \leq x) = F_X(x) .
\end{align}
We also assume $\{Y_i\}_{i\ge1}$  to be a renewal process with cdf $F_Y(x)$, i.e.,
\begin{align}
\mathbb{P}(Y_{i} \leq x) = F_Y(x).
\end{align}
We denote $\alpha_X = 1/\mu$ and $\sigma_X^2$ to be the mean and variance associated with $F_X$. Similarly let  $\alpha_Y = 1/\lambda$ and $\sigma_Y^2$ be the mean and variance associated with $F_Y$. We let $\rho = \lambda/\mu$ and assume that $\rho < c$. Denote by $F_X^{*}(s) = \int_0^\infty e^{-sx}dF_X(x)$ and $F_Y^{*}(s)$ the Laplace-Stieltjes transform ({\it LST}) of $F_X$ and $F_Y$ with $s\ge 0.$

In our paper, we consider various inter-server and inter-user distance distributions, including exponential, deterministic, uniform and hyperexponential.

\subsection{Allocation policies}
One of our goals is to analyze the performance of various request allocation policies using expected request distance as a performance metric. We define various allocation policies as follows.
\begin{itemize}
\item{\bf Unidirectional Gale-Shapley ({\it UGS}):}
In UGS, each user simultaneously emits a ray to their right. Once the ray hits an unallocated server $s$, the user is allocated to $s$.
\item\noindent{\bf Move To Right ({\it MTR}):}
In MTR, starting from the left, each user is allocated sequentially to the nearest available server to its right.
\np{\item{\bf Nearest Neighbor ({\it NN})  \cite{Stuart10}:} 
In this matching, starting from the left,  each user is allocated sequentially to the nearest available server. This policy can be viewed as the bidirectional version of MTR policy.}
\item\noindent{\bf Gale-Shapley ({\it GS}) \cite{Gale62}:}
In this matching, each user selects the nearest server and each server selects its nearest user. Remove reciprocating pairs, and continue.
\item\noindent{\bf Optimal Matching:}
This matching minimizes average request distance among all feasible allocation policies.
\end{itemize}

%% file: 04-queue.tex
\begin{figure}[]
\centering
\begin{minipage}{.35\textwidth}
\centering
\includegraphics[width=\linewidth]{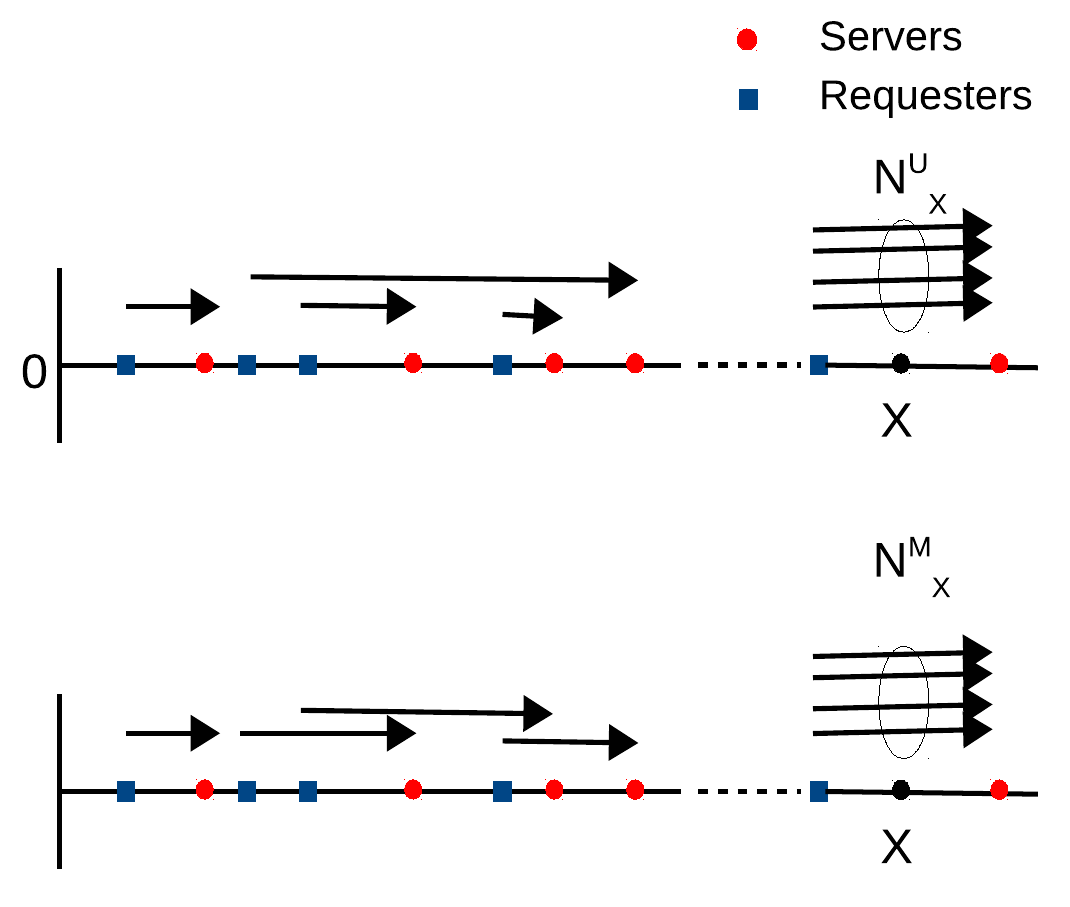}
\end{minipage}
\caption{Allocation of users to servers on the one-dimensional network. Top: UGS, Bottom: MTR allocation policy.}
\label{1d-grid-poisson}
\end{figure} 

In this Section, we establish the equivalence of UGS and MTR w.r.t number of requests that traverse a point and  expected request distance. Define $N_{x}^{P}$ and  $D_i^P$ to be random variables for the number of requests that traverse point $x \in {\mathcal L}$ and distance between user $i$ and its allocated server under policy $P$, respectively. Thus $N_{x}^{U}$ and  $N_{x}^{M}$ denote the number of requests that traverse point $x \in {\mathcal L}$ under UGS and MTR, respectively, as shown in Figure \ref{1d-grid-poisson}. Consider the following definition of busy cycle in a service network.

\begin{define}
A busy cycle for a policy P is an interval $I = [a,b] \subset {\mathcal L}$ such that $\exists\; i,j $ with $r_i = a, s_j = b$ for which $N_{x}^{P} > 0, \forall x \in I$ and $N_{x}^{P} = 0$ for $x = a-\epsilon$ and $x = b+\epsilon$  with $\epsilon$ being an infinitesimal positive value.
\end{define}
We have the following theorem.
\begin{theorem}\label{lm:queue_mtr}
$N_{x}^{U} = N_{x}^{M}, x\ge 0.$
\end{theorem}
\begin{proof}
Due to the unidirectional nature of matching, both UGS and MTR have the same set of busy cycles. Denote ${\mathcal I}$ as the set of all busy cycles in the service network. In the case when $x \in {\mathcal L}\setminus \bigcup\limits_{I\in{\mathcal I}} I$ we already have $N_{x}^{U} = N_{x}^{M} = 0.$ Let us now consider a busy cycle $I^U = [a^U,b^U]$ under UGS policy. Let $x\in I^U.$ Let $L^U_{x,R} = |\{r_i|a^U\le r_i\le x\}|$ and $L^U_{x,S} = |\{s_j|a^U\le s_j\le x\}|.$ $N_{x}^{U} = L_{x,R}^{U} - L_{x,S}^{U}.$ Similarly define $L_{x,R}^{M}$ and $L_{x,S}^{M}$ for MTR policy. Clearly $N_{x}^{M} = L_{x,R}^{M} - L_{x,S}^{M}.$ As both policies have the same set of busy cycles we have $L_{x,R}^{U} = L_{x,R}^{M}$ and $L_{x,S}^{U} = L_{x,S}^{M}.$ Thus we get
\begin{align}
N_{x}^{U} = N_{x}^{M},\; x \in \mathbb{R}^{+},
\end{align}
\end{proof}
\begin{corollary}
$\mathbb{E}[D^U] = \mathbb{E}[D^M]$ i.e. the expected request distances are the same for both UGS and MTR under steady state.
\end{corollary}
\begin{proof}
Under steady state both $N_{x}^{U}$ and $N_{x}^{M}$  converge to a random variable. Applying Little's law we have  $\mathbb{E}[D^U] = \mathbb{E}[D^M].$
\end{proof}
\begin{remark}
Note that Theorem \ref{lm:queue_mtr} applies to any inter-server or inter-user distance distribution. It also applies to the case where servers have capacity $c>1.$
\end{remark}
\begin{remark}\label{remark:variance}
Although MTR and UGS are equivalent w.r.t. the expected request distance, MTR tends to be fairer, i.e., has low variance\footnote{\np{It is well known in queueing theory that among all service disciplines  the variance of the waiting time is minimized under FCFS policy for Poisson arrivals and exponential service times \cite{Kingman62}. In Section \ref{sec:mm1} we show that MTR maps to a temporal FCFS queue.}} w.r.t. request distance.
\end{remark}

%% file: 05-mm1.tex

\begin{table*}[]
\center
\begin{tabular}{|c| c| c| c|}
\hline
 {\bf Distribution}&{\bf Parameters}&$\pmb{F_X(x)}$&$\pmb{\mathcal{B}(x)}$\\
\hline
Exponential&$\mu$: rate&$1-e^{-\mu x}$&$\frac{1}{\lambda}\left[1-e^{-\lambda x}\right]-\frac{1}{\lambda+\mu}\left[1-e^{-(\lambda+\mu) x}\right]$\\
&&&\\
\hline
Uniform&$b:$ maximum value&$x/b,\quad 0 \leq x \leq b$&$\frac{1}{\lambda^2b}\left[1-e^{-\lambda b}\right]-\frac{e^{-\lambda x}}{\lambda}$\\
\hline
Deterministic&$d_0:$ constant&$1, \quad x \ge  d_0$&$\frac{e^{-\lambda d_0}-e^{-\lambda x}}{\lambda}$\\
\hline
Hyper&$l$: order&$1-\sum\limits_{j=1}^{l}p_{j}e^{-\mu_{j}x}$&$\frac{1}{\lambda}\left[1-e^{-\lambda x}\right]-\sum\limits_{j=1}^{l}\frac{p_j}{\lambda+\mu_{j}}\left[1-e^{-(\lambda+\mu_{j}) x}\right]$\\
-exponential&$p_{j}:$ phase probability &&\\
&$\mu_{j}:$ phase rate&&\\
\hline
\end{tabular}
\caption{Properties of specific inter-server distance distributions.}
\label{tbltraff}
\end{table*}

In this section, we characterize request distance statistics under unidirectional policies when both users and servers are distributed according to two independent Poisson processes. We first analyze MTR as follows.
%
%
\subsection{MTR}

Under this allocation policy, the service network can be modeled as a bulk service M/M/1 queue. A bulk service M/M/1 queue provides service to a  group of $c$ or fewer customers. The server serves a bulk of at most $c$ customers whenever it becomes free. Also customers can join an existing service if there is room which is an example of accessible batch. In Section \ref{sec:mg1} we describe the notion of accessible batches in greater detail. The service time for the group is exponentially distributed and customer arrivals are described by a Poisson process. The distance between two consecutive users in the service network can be thought of as inter-arrival time between customers in the bulk service M/M/1 queue. The distance between two consecutive servers maps to a bulk service time.

Having established an analogy between the service network and the bulk service M/M/1 queue, we now define the state space for the service network. Consider the definition of $N_x$ as the number of requests\footnote{We drop the superscript $(M)$ for brevity.} that traverse point $X \in L$ under MTR. In steady state, $N_x$ converges to a random variable $N$ provided $\lambda < c\mu$. Let $\pi_k$ denote $\texttt{Pr}[N = k]$ with $k\ge 0$.

%
Following the procedure in Section 4.2.1 of \cite{Gross98}, we obtain the steady state probability vector $\pi = [\pi_i, i\ge0].$ 
In the service network, request distance corresponds to the sojourn time in the bulk service M/M/1 queue. By applying Little's formula, we obtain the following expression for the expected request distance
\begin{align}
\mathbb{E}[D] = \frac{r_0}{\lambda(1-r_0)}, \label{eq:c2}
\end{align}
\noindent where $r_0$ is the only root in the interval $(0,1)$ of the following equation (with $r$ as the variable)
\begin{align}\label{eq:polymm1}
\mu r^{c+1} - (\lambda+\mu)r + \lambda = 0.
\end{align}

\subsubsection{When server capacity: $c =1$}
When $c=1,$ $r_0 = \rho$ is a solution of \eqref{eq:polymm1}. Thus we can evaluate the expected request distance as 
\begin{align}
\mathbb{E}[D] = \frac{\rho}{\lambda(1-\rho)} = \frac{1}{\mu-\lambda}. \label{eq:c12}
\end{align}
Note that, when server capacity is one, the service network can be modeled as an M/M/1 queue. In such a case,   \eqref{eq:c12} is the mean sojourn time for an M/M/1 queue.


\subsection{UGS}
When both users and servers are Poisson distributed and servers have unit capacity, the request distance in UGS has the same distribution as the busy cycle in the corresponding Last-Come-First-Served Preemptive-Resume ({\it LCFS-PR}) queue having the density function \cite{Abadi17}
\begin{align}
f_{D^U}(x) = \frac{1}{x\sqrt{\rho}}e^{(\lambda+\mu)x}I_1(2x\sqrt{\lambda\mu}),\; x > 0,
\end{align}
\noindent where $\rho = \lambda/\mu$ and $I_1$ is the modified Bessel function of the first kind. Thus the expected request distance is equivalent to the average busy cycle duration in a LCFS-PR queue given by $1/(\mu-\lambda)$ \cite{Abadi17}.

When servers  have capacities $c>1$ it is difficult to characterize the expected request distance explicitly. However, by Theorem \ref{lm:queue_mtr}, the expected request distance under UGS is the same as that of MTR given by \eqref{eq:c2}.

%% file: 06-mg1.tex

\begin{figure}[]
\centering
\includegraphics[width=0.4\linewidth]{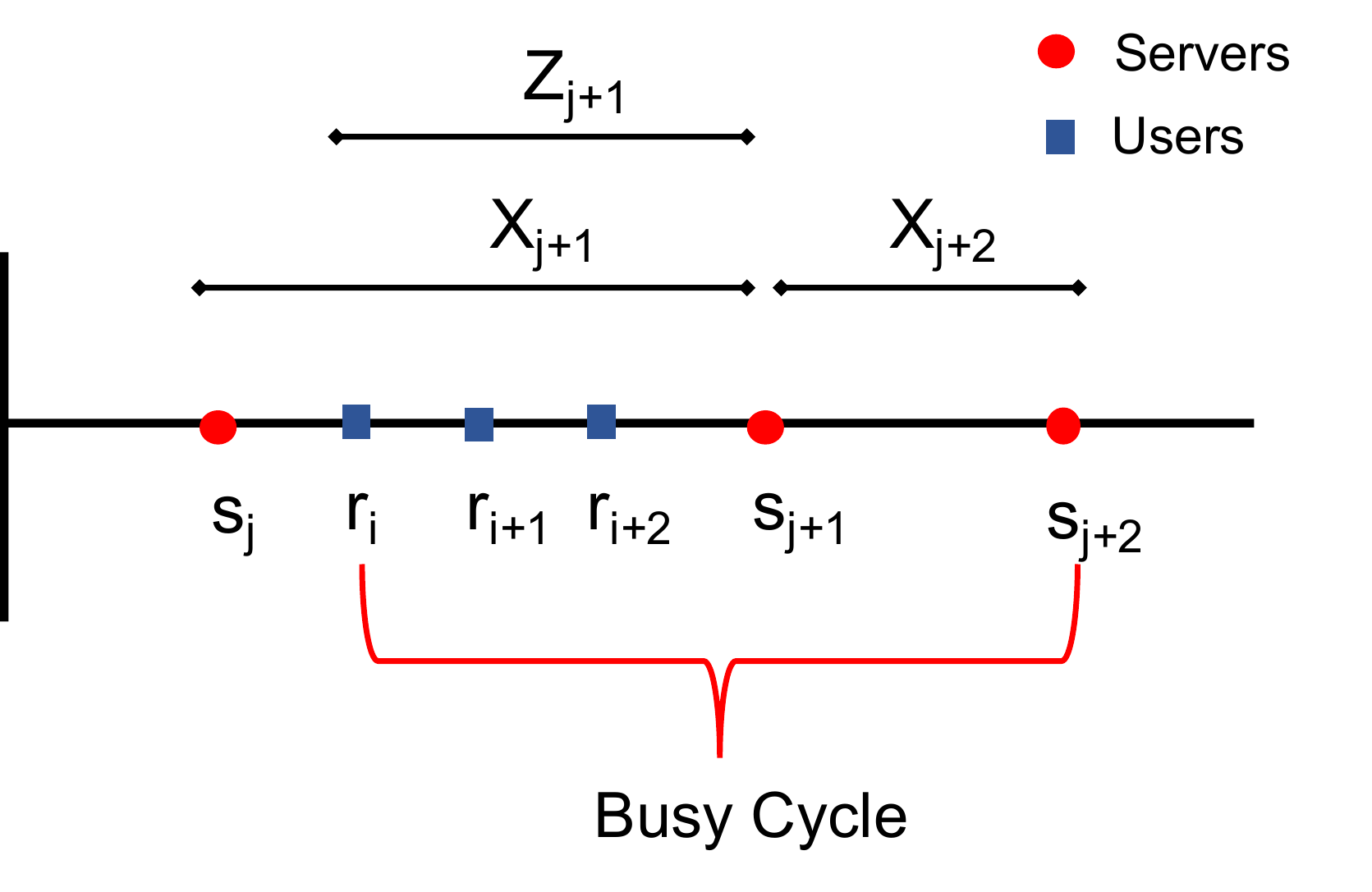}
\vspace{-0.1in}
\caption{Allocation of users to servers under MTR policy.}
\label{1d-grid-nonpoisson}
\vspace{-0.1in}
\end{figure} 
We now derive expressions for the expected request distance when either users or servers are distributed according to a Poisson process and the other by renewal process. 
\subsection{Notion of exceptional service and accessible batches}
We discuss the notion of exceptional service and accessible batches applicable to our service network as follows. Consider a service network with $c=2$  as shown in Figure \ref{1d-grid-nonpoisson}. Consider a user $r_i$. Let $s_j$ be the server immediately to the left of $r_i.$ We assume all users prior to  $r_i$ have already been allocated to servers $\{s_k, 1\le k\le j\}$.  MTR allocates both $r_i$ and $r_{i+1}$ to $s_{j+1}$ and allocates $r_{i+2}$ to $s_{j+2}.$ We denote $[r_i,s_{j+2}]$ as a busy cycle of the service network. We have the following queueing theory analogy.

User $r_i$ can be thought of as the first customer in a queueing system that initiates a busy period while $r_{i+1}$ sees the system busy when it arrives. Because only $r_i$ is in service at the arrival of $r_{i+1}$, $r_{i+1}$ enters service with $r_i$ and the two customers form a batch of size 2. and depart at time $s_{j+1}$.  This is an example of an {\it accessible batch} \cite{Goswami11}. An accessible batch admits subsequent arrivals, while the service is on, until the server capacity $c$ is reached.

The service time for the batch, $r_i,r_{i+1}$, is described by the random variable $Z_{j+1}$ which is different or {\it exceptional} when compared to service times of successive batches such as the one consisting of $r_{i+2}$. The service time for the second batch is $X_{j+2}.$ Note that, $Z_{j+1}$ only depends on $X_{j+2}$ and $Y_{i+2}.$ Thus  when either $X_{j+2}$ or $Y_{i+2}$ is described by a Poisson process and the other by renewal process, $Z_{j+1}$ converges to a random variable $Z$ under steady state conditions. Denote $F_Z(x)$ and $f_Z(x)$ as the distribution and density functions for the random variable $Z$. Thus the service network can be mapped to an  exceptional service with accessible batches queueing ({\it ESABQ}) model.  We formally define ESABQ as follows.\\

\noindent{\bf ESABQ:} {\it Consider a queueing system where customers are served in batches of maximum size $c$. A customer entering the queue and finding fewer than $c$ customers in the system joins the current batch and enters service at once, otherwise it joins a queue. After a batch departs leaving $k$ customers in the buffer, $\min(c,k)$ customers form a batch and enter service immediately. There are two different service times cdfs, $F_Z(x)$ (exceptional batch) with mean $\alpha_Z=1/\mu_Z$ and $F_X(x)$ (ordinary batch) with mean $\alpha_X=1/\mu$. A batch is exceptional if its oldest customer entered an empty system, otherwise it is a regular batch. When the service time expires, all customers in the server depart at once, regardless of the nature of the batch (exceptional or regular).}

\subsubsection{Evaluation of the distribution function: $F_Z(x)$}\label{eq:sub_ge}
In this Section, we compute explicit expressions for the distribution function $F_Z(x)$ applicable to our service network.\\

\noindent{\bf When $\pmb{F_X(x)\sim} \texttt{Expo}\pmb{(\mu)}$:} In this case, we invoke the memoryless property of the exponential distribution $F_X$. Thus the exceptional distribution, $F_Z$, is
\begin{align}
F_Z(x) = F_X(x) = 1-e^{-\mu x}, x\ge0.
\end{align}

\noindent{\bf When $\pmb{F_Y(x)\sim} \texttt{Expo}\pmb{(\lambda)}$:}
Using the memoryless property of $F_Y$, $F_Z$ can be computed as  
{\small
\begin{align}
F_Z(x) &= \text{Pr}(X-Y<x|Y<X)= \text{Pr}(X-Y<x|X-Y >0) = \frac{\text{Pr}(X-Y<x) - \text{Pr}(X-Y<0)}{1 - \text{Pr}(X-Y<0)}\nonumber\\
&= \frac{D_{XY}(x) - D_{XY}(0)}{1-D_{XY}(0)},\;x\ge0,\label{eq:first_busy}
\end{align}
}
\noindent where $D_{XY}(x)$ is the distribution of the random variable $X-Y$ (also known as difference distribution). $D_{XY}(x)$ can be expressed as
{\small
\begin{align}
D_{XY}(x) &= \text{Pr}(X-Y\le x) = \int_0^\infty\text{Pr}(X-y\le x)\text{Pr}(Y=y) dy = \int_0^{\infty} F_X(x+y)\lambda e^{-\lambda y} dy = \int_x^{\infty} F_X(z)\lambda e^{-\lambda (z-x)} dz\nonumber\\ 
&= \lambda e^{\lambda x}\left[\int_0^{\infty} F_X(z) e^{-\lambda z} dz - \int_0^{x} F_X(z) e^{-\lambda z} dz \right] = \lambda e^{\lambda x}\left[\mathcal{A}(F_X)-\mathcal{B}(x)\right],\label{eq:diff_comm}
\end{align}
}

\noindent where $\mathcal{A}$ is the Laplace Transform operator on the function $F_X$ and  $\mathcal{B}(x)$ is denoted by
$$\mathcal{B}(x) = \int_0^{x} F_X(z) e^{-\lambda z} dz$$ 
Clearly $\mathcal{B}(0) =0.$ Thus combining \eqref{eq:first_busy} and \eqref{eq:diff_comm} yields

\begin{align}
F_Z(x) &= \frac{\lambda e^{\lambda x}\left[\mathcal{A}(F_X)-\mathcal{B}(x)\right] - \lambda\mathcal{A}(F_X)}{1-\lambda\mathcal{A}(F_X)},\label{eq:ge_final}\\
f_Z(x) &= \frac{\lambda^2 e^{\lambda x}\left[\mathcal{A}(F_X)-\mathcal{B}(x)\right]-\lambda F_X(x)}{1-\lambda\mathcal{A}(F_X)},\\
\alpha_Z &= \int_{0}^{\infty}xf_Z(x)dx,\;\sigma_Z^2 = \left[\int_{0}^{\infty}x^2f_Z(x)dx\right] - \alpha_Z^2.
\end{align}

Expressions for $\mathcal{B}(x)$ are presented in Table \ref{tbltraff}. We can evaluate $\mathcal{A}(F_X)$ by setting $\mathcal{A}(F_X) = B(\infty).$ Detailed derivations are relegated to Appendix \ref{app-Ge-distribs}.

\subsection{General requests and Poisson distributed servers ({\it GRPS})}\label{sub:grps1}
From our discussion in Section \ref{eq:sub_ge}, it is clear that when servers are distributed according to a Poisson process, the exceptional service time distribution equals the regular batch service time distribution. In such a case we have the following queueing model. \\

\noindent{\it Under GRPS, inter-arrival times and batch service times are, respectively, arbitrarily and exponentially distributed. Before initiating a service, a server finds the system in any of the following conditions. (\romannumeral 1) $1\le n\le c-1$ and (\romannumeral 2) $n\ge  c.$  Here $n$ is the number of customers in the waiting buffer. For case (\romannumeral 1) the server provides service to all $n$ customers and admits subsequent arrivals until $c$ is reached. For case (\romannumeral 2) the server takes $c$ customers with no admission for subsequent customers arriving within its service time.}\\

In such a case ESABQ can directly be modeled as a special case of a renewal input bulk service queue with accessible and non-accessible batches proposed in \cite{Goswami11}  with parameter values $a=1$ and $d = b = c.$ Let $N_s$ and $N_q$ denote random variables for numbers of customers in the system and in the waiting buffer respectively for ESABQ under GRPS. We borrow the following definitions from \cite{Goswami11}.
\begin{align}\label{eq:probs}
P_{n,0} &= \Pr[N_s = n]; 0\le n\le c-1,\;P_{n,1} = \Pr[N_q = n]; n\ge0.
\end{align}
Using results from \cite{Goswami11} we obtain the following expressions for equilibrium queue length probabilities.
\begin{align}\label{eq:probs2}
P_{0,1} &= \frac{C}{\mu}\bigg[\frac{r_0^{c-1}-r_0^c}{1-r_0^c}+\frac{1}{r_0}-1\bigg], \; P_{n,1} = \frac{Cr_0^{n-1}(1-r_0)}{\mu(1-r_0^c)}; n\ge1,
\end{align}
where $0<r_0<1$ is the real root of the equation $r=F^*_Y(\mu-\mu r^c)$  and $C$ is the normalization constant\footnote{The normalization constant $C$ derived in \cite{Goswami11} is incorrect. The correct constant for our case is given in \eqref{eq:probs_C}.} given by
{\footnotesize
\begin{align}\label{eq:probs_C}
C = \lambda\bigg[\frac{1-\omega^{c}}{1-\omega} + \frac{1}{1-r_0} - \frac{\omega(r_0-F^*_Y(\mu))}{r_0^c(1-r_0\omega)}\bigg(\frac{1-r_0^c}{1-r_0}-r_0^{c-1}\frac{1-w^c}{1-w}\bigg)\bigg]^{-1},
\end{align}
}
\noindent with $\omega = 1/F^*_Y(\mu)$. We then derive the expected queue length as
\begin{align}\label{eq:qlen}
\mathbb{E}[N_q] &= \sum\limits_{n=0}^{\infty}nP_{n,1} = \sum\limits_{n=1}^{\infty}n\frac{Cr_0^{n-1}(1-r_0)}{\mu(1-r_0^c)} = \frac{C(1-r_0)}{\mu(1-r_0^c)}\sum\limits_{n=1}^{\infty}nr_0^{n-1} = \frac{C}{\mu(1-r_0^c)(1-r_0)}.
\end{align}

Applying Little's law and considering the analogy between our service network and ESABQ we obtain the following expression for the expected request distance.
\begin{align}
\mathbb{E}[D] = \frac{C}{\lambda\mu(1-r_0^c)(1-r_0)} + \frac{1}{\mu}.
\end{align} 
\subsection{Poisson distributed requests and general distributed servers ({\it PRGS})}\label{sub:prgs1}
\input{06-mg1-bulk}

%% file: 06-mg1-bulk.tex

As discussed in Section \ref{eq:sub_ge}, if servers are placed on a $1$-d line according to a renewal process with requests being Poisson distributed, the service time distribution for the first batch in a busy period differs from those of subsequent batches. Below we derive expressions for queue length distribution and expected request distance for ESABQ under PRGS.

\subsubsection{Queue length distribution}\label{sub:ql_prgs}
We use a supplementary variable technique to derive the queue length distribution for ESABQ under PRGS as follows.

Let $L(t)$  be the number of customers at time $t\geq 0$, $R(t)$ the residual service time at time $t\geq 0$ (with $R(t)=0$ if $L(t)=0$), and $I(t)$ the type of service at time $t\geq 0$ with $I(t)=1$ (resp. $I(t)=2$) if exceptional (resp. ordinary) service time.

Let us write the Chapman-Kolmorogov equations for the Markov chain $\{(L(t), R(t), I(t)),\, t\geq 0\}$.

For $t\geq 0$, $n\geq 1$, $x>0$, $i=1,2$ define 
\begin{eqnarray*}
p_t(n,x;i) = \P(L(t)=n, R(t)<x, I(t)=i) \quad \hbox{and} \quad p_t(0) = \P(L(t)=0). 
\end{eqnarray*}
Also, define for $x>0$, $i=1,2$,
\[
p(n,x;i)=\lim_{t\to\infty} p_t(n,x;i) \quad \hbox{and} \quad  p(0)=\lim_{t\to\infty}p_t(0).
\]
By analogy with the analysis for the M/G/1 queue we get 
\[
\frac{\partial}{\partial t} p_t(0)=-\lambda p_t(0)+ \sum_{k=1}^c \frac{\partial }{\partial x} p_t(k,0;1) + \sum_{k=1}^c \frac{\partial}{\partial x} p_t(k,0;2),
\]
so that, by letting $t\to\infty$,
\begin{equation}
\lambda p(0)=\sum_{k=1}^c \left(\frac{\partial }{\partial x} p(k,0;1) +\frac{\partial}{\partial x} p(k,0;2)\right).
\label{eq:p0}
\end{equation}

With further simplification (See Appendix \ref{sub:ck}), for $n\geq 1, x>0$ we get
\begin{align}
&\frac{\partial}{\partial x} g(n,x)-\lambda g(n,x)- \frac{\partial}{\partial x} g(n,0)+
\lambda g(n-1,x){\bf 1}(n\geq 2) +\lambda p(0)F_Z(x){\bf 1}(n=1)+ F_X(x) \frac{\partial}{\partial x} g(n+c,0) = 0,
\label{eq:21}
\end{align}
\noindent where $g(n,x)=p(n,x;1)+p(n,x;2)$ for $n\geq 1$, $x>0$. Introduce 
\[
G(z,s):=\sum_{n\geq 1}z^n \int_0^\infty e^{-sx}  g(n,x) dx \quad \forall |z|\leq 1, \,s\geq 0. 
\]
Denote by $F_{Z}^*(s)=\int_0^\infty e^{-sx} dF_{Z}(x)$ the LST of $F_{Z}$ for $s\geq 0$. Note that
\[
\int_0^\infty e^{-s x} F_{Z\hbox{\tiny{or}}X}(x)dx=\frac{F^*_{Z\hbox{\tiny{or}}X}(s)}{s}, \quad \forall s>0.
\]
Multiplying both sides of (\ref{eq:21}) by $z^n e^{-sx}$, integrating over $x\in [0,\infty)$ and summing over all $n\geq 1$, yields
\begin{align}
s\left(\lambda(1-z)-s\right)G(z,s) = &\lambda zp(0)F^*_Z(s)- \sum_{n\geq 1}z^n \frac{\partial}{\partial x} g(n,0)+F^*_X(s)\sum_{n\geq 1} z^n \frac{\partial}{\partial x} g(n+c,0))
\label{eq:22}
\end{align}
where $\lambda p(0)= \sum_{k=1}^c \frac{\partial}{\partial x}g(k,0)$ from (\ref{eq:p0}). We have
\begin{align}
\frac{1}{z^c}\sum_{n\geq 1} z^{n+c} \frac{\partial}{\partial x} g(n+c,0)) 
= \frac{1}{z^c}\sum_{n \geq 1} z^n\frac{\partial}{\partial x} g(n,0)-\frac{1}{z^c}H(z)
\end{align}
where $H(z)=\sum_{k=1}^c  z^k a_k$ with $a_k:=\frac{\partial}{\partial x} g(k,0),$ for $k=1,\ldots,c$. Introducing the above into (\ref{eq:22}) gives
{\small
\begin{align}
s\left(\lambda(1-z)-s\right)G(z,s)=&\left(\frac{F^*_X(s)}{z^c}-1\right)\Psi(z) -F^*_X(s)\frac{H(z)}{z^c}+\lambda zp(0)F^*_Z(s)\label{eq:23}
\end{align}
}

\noindent where $\Psi(z):=\sum_{n\geq 1}z^n \frac{\partial}{\partial x}g(n,0)$.
Since $G(z,s)$ is well-defined for $|z|\leq 1$ and $s\geq 0$, the r.h.s. of (\ref{eq:23}) must vanish when $s=\lambda(1-z)$.   
This gives the relation
\[
\Psi(z) = \frac{z^c}{z^c- F^*_X(\theta(z))} \left[  -F^*_X(\theta(z))\frac{H(z)}{z^c}+\lambda zp(0)F^*_Z(\theta(z))\right]
\]
with $\theta(z)=\lambda(1-z)$ and $|z|\leq 1$. Introducing the above in (\ref{eq:23}) gives
\begin{align}
s\left(\lambda(1-z)-s\right)G(z,s)= -F^*_X(s)\frac{H(z)}{z^c}+\lambda zp(0)F^*_Z(s) +\frac{F^*_X(s)-z^c}{z^c-F^*_X(\theta(z))}\left[\lambda zp(0)F^*_Z(\theta(z)) -F^*_X(\theta(z))\frac{H(z)}{z^c}\right].
\label{eq:24}
\end{align}
Let $N(z)$ be the $z$-transform of the stationary number of customers in the system. \np{Integrating by part, we get for $n\geq 1$,
\[
s\int_0^\infty e^{-sx} g(n,x)dx= \int_0^\infty e^{-sx}dg(n,x),
\]
so that 
\begin{equation}
\label{Nz}
\lim_{s\to\infty}s\int_0^\infty e^{-sx} g(n,x)dx = \lim_{s\to 0}  \int_0^\infty e^{-sx}dg(n,x)= \int_0^\infty dg(n,x)=g(n,\infty),
\end{equation}
where the interchange between the limit and the integral sign is justified by the bounded convergence theorem. Therefore, 
\begin{eqnarray}
N(z)&=&\sum_{n\geq 1} z^n g(n,\infty)+p(0)\nonumber\\
&=& \sum_{n\geq 1} z^n \lim_{s\to\infty}s\int_0^\infty e^{-sx} g(n,x)dx\quad \hbox{from } (\ref{Nz})\nonumber \\
&=&\lim_{s\to 0} s G(z,s)+p(0), \label{Nz-2}
\end{eqnarray}
where the  interchange  between the summation over $n$ and the integral sign is again justified by the bounded convergence theorem.}
Letting now $s\to 0$ in (\ref{eq:24}) and using (\ref{Nz-2}), gives
\begin{align}
\theta(z)N(z)= \frac{1-z^c}{z^c -F^*_X(\theta(z))}&\left[-F^*_X(\theta(z))\frac{H(z)}{z^c}+\lambda zp(0)F^*_Z(\theta(z))\right]- \frac{H(z)}{z^c}+\lambda  p(0). \label{eq:25} 
\end{align}

By noting that $\lambda p(0)=\sum_{k=1}^c a_k$ (cf. (\ref{eq:p0})), Eq. (\ref{eq:25}) can be rewritten as

\begin{align}
N(z)=\frac{1}{\theta(z)}&\bigg(\frac{z(1-z^c)}{z^c-F^*_X(\theta(z))} \sum_{k=1}^c a_k\left[F^*_Z(\theta(z))-z^{k-c-1}F^*_X(\theta(z))\right] +\sum_{k=1}^c a_k (1-z^{k-c})\bigg).
\label{eq:30a}
\end{align}

The r.h.s. of (\ref{eq:30a}) contains $c$ unknown constants $a_1,\ldots, a_c$ yet to be determined. Define $A(z)=F^*_X(\theta(z))$. It can be shown that $z^c-A(z)$ has $c-1$ zeros inside and one on the unit circle, $|z|= 1$ (See Appendix \ref{app-mg1-rouche}). Denote by $\xi_1,\ldots,\xi_q$ the $1\leq q\leq c$ distinct zeros of $z^c -A(z)$ in $\{|z|\leq 1\}$, with multiplicity $n_1,\ldots,n_q$, respectively, with $n_1+\cdots+n_q=c$. Hence, 
\[
z^c- F_X^*(k(z))= \gamma \prod_{i=1}^q (z-\xi_i)^{n_i}.
\]
Since $z^c -A(z)$ vanishes when $z=1$ and that $\frac{d}{dz}(z^c -A(z))|_{z=1}=c-\rho>0$, we conclude that $z^c -A(z)$ has one zero of multiplicity one at $z=1$.

Without loss of generality assume that $\xi_q=1$ and let us now focus on the zeros $\xi_1,\ldots,\xi_{q-1}$. When $z=\xi_i$, $i=1,\ldots,q-1$, 
the term $F^*_Z(\theta(z))-z^{k-c-1}F^*_X(\theta(z))$ in (\ref{eq:30a}) must have a zero of multiplicity (at least) $n_i$ since $N(\xi_i)$ is well defined. This gives $c-1$ linear equations to
be satisfied by $\xi_1,\ldots, \xi_q$. In the particular case where all zeros have multiplicity one (see Appendix \ref{sub:rem1}), namely $q=c$, these $c-1$ equations are
\begin{equation}
\label{eq:1-c-1}
\sum_{k=1}^c a_k\left[F^*_Z(\theta(\xi_i))- \xi_i^{k-c-1}F^*_X(\theta(\xi_i))\right]=0, \,\, i=1,\ldots, c-1.
\end{equation}
With $U(z):=F^*_Z(\theta(z))/F^*_X(\theta(z))$ (\ref{eq:1-c-1}) is equivalent to
\begin{equation}
\label{eq:30}
\sum_{k=1}^c a_k\left[U(\xi_i)- \xi_i^{k-c-1})\right]=0, \,\, i=1,\ldots, c-1,
\end{equation}
since $F^*_X(\theta(\xi_i))\not=0$ for $i=1,\ldots,c-1$ ($F^*_X(\theta(\xi_i))=0$ implies that $\xi_i$=0 which contradicts that $\xi_i$ a zero of $z^c-F^*_X(\theta(z))$
since $F^*_X(\theta(0))= F^*_X(\lambda)>0$).
Eq. (\ref{eq:30a}) can be rewritten as
{\footnotesize
\begin{equation}
N(z)
=\frac{ \sum_{k=1}^c a_k\left[z^c-z^k+ z(1-z^c) F^*_Z(\theta(z))- (1-z^k)F^*_X(\theta(z))\right]}{\theta(z) (z^c-F^*_X(\theta(z))}.
\label{value:Nz}
\end{equation}
}
A $c$-th equation is provided by the normalizing condition $N(z)=1$. Since the numerator and denominator in (\ref{value:Nz}) have a zero of order $2$ at $z=1$, 
differentiating twice the numerator and the denominator w.r.t $z$ and letting $z=1$ gives
\begin{equation}
\label{eq:c}
\sum_{k=1}^c a_k (c(1+\rho_z)-\rho k) = \lambda (c-\rho),
\end{equation}

\noindent where $\rho_z = \lambda\alpha_Z.$ We consider few special cases of the model in Appendix \ref{app:verify} and verify with the expressions of queue length distribution available in the literature.


%
%

%

\subsubsection{Expected request distance}
From \eqref{value:Nz} the expected queue length is
{\footnotesize
\begin{align}
&\overline{N}=\frac{d}{dz} N(z)\Big|_{z=1}\nonumber\\
&= \frac{1}{2\lambda(c-\rho)^2}\sum_{k=1}^c a_k\Biggl[
\lambda^2 \sigma^{(2)}_Zc(c-\rho)+ \lambda^2 \sigma^{(2)}_X c(1+\rho_z-k)+(ck(c-k)+k(k-1)\rho-c(c-1))\rho+2c^2 \rho_z -c(c+1)\rho_z\rho\Biggr],
\label{eq:mean-QL}
\end{align}
}

\noindent where $\sigma^{(2)}_Z$ and $\sigma^{(2)}_X$ are the second order moments of distributions $F_Z$ and $F_X$ respectively. Again by applying Little's law and considering the analogy between our service network with ESABQ we get the following expression for the expected request distance.
\begin{align}
\mathbb{E}[D] = \overline{N}/\lambda.
\end{align} 

%% file: 07-generalizations.tex

In this section we describe generalizations of models and results for unidirectional allocation policies. We first consider the case when inter-user and inter-server distances both have general distributions.

\begin{figure}
\centering
\includegraphics[width=0.4\textwidth]{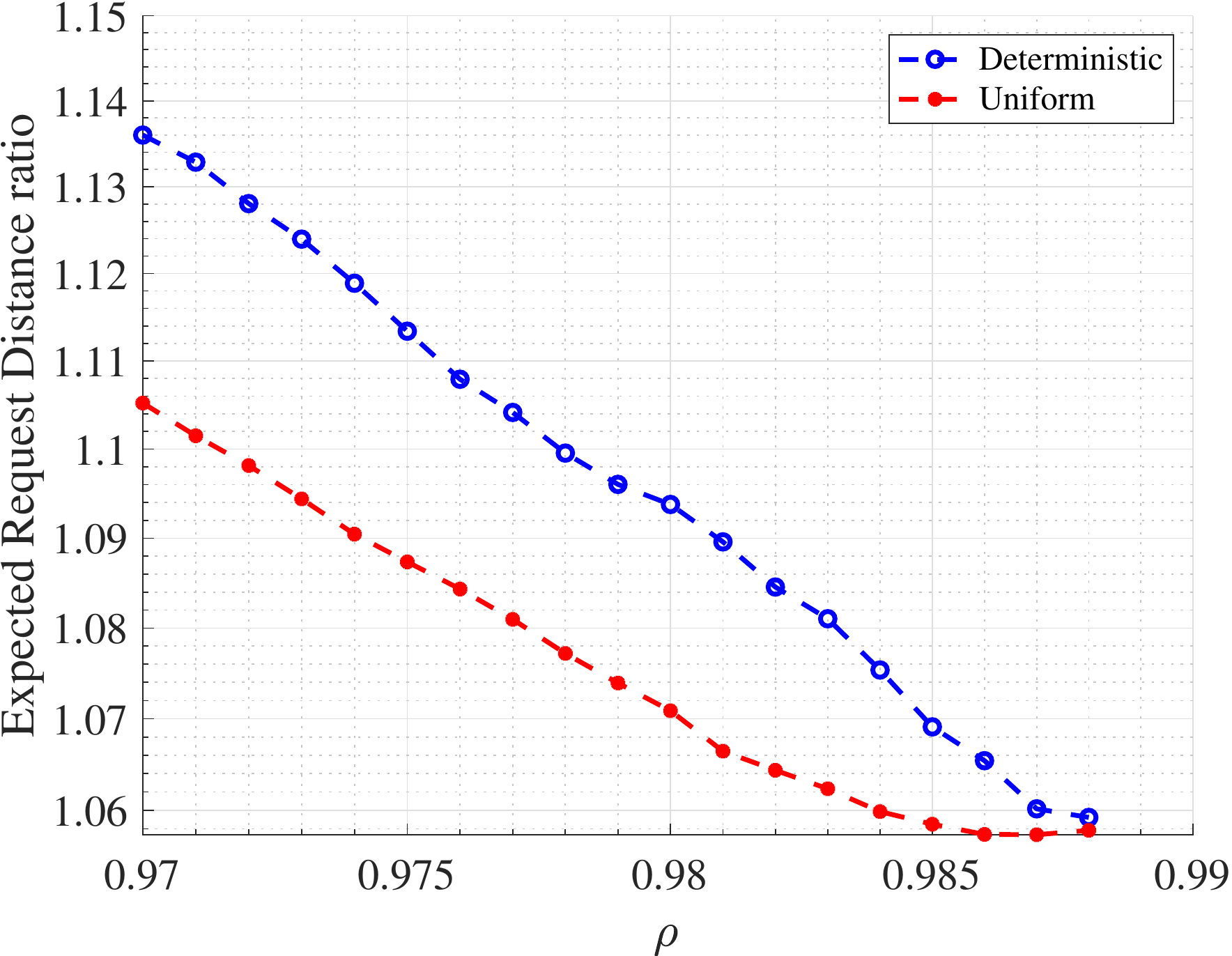}
\caption{The plot shows the ratio $\mathbb{E}[D]/\overline{D}_{s}$ for deterministic and uniform inter-server distance distributions.}
\label{1d-grid-temp-vs-spat-deterministic}
\vspace{-0.15in}
\end{figure}
\subsection{Heavy traffic limit for general request and server spatial distributions}
Consider the case when the inter-user and inter-server distances each are described by general distributions. We assume server capacity, $c=1$. As $\rho \to 1$, we conjecture that the behavior of MTR approaches that of the G/G/1 queue. One argument in favor of our conjecture is the following. As $\rho \to 1$, the busy cycle duration tends to infinity. Consequently, the impact of the exceptional service for the first customer of the busy period on all other customers diminishes to zero as there is an unbounded increasing number of customers served in the busy period.

It is known that in heavy traffic waiting times in a G/G/1  queue  are exponential distributed and the mean sojourn time is given by $\alpha_X+[(\sigma_X^2+\sigma_Y^2)/2\alpha_Y(1-\rho)]$  \cite{Gross98}. We expect the expected request distance to exhibit similar behavior. Thus we have the following conjecture.

\begin{conjecture}
At heavy traffic i.e. as $\rho \to 1$, the expected request distance for the G/G/1 spatial system with $c=1$ is given by
\begin{align}
\mathbb{E}[D] = \alpha_X + \frac{\sigma_X^2+\sigma_Y^2}{2\alpha_Y(1-\rho)}.\label{eq:temp_g_g_1}
\end{align}
\end{conjecture}

Denote by $\overline{D}_{s}$ the average request distance as obtained from simulation. We plot the ratio $\mathbb{E}[D]/\overline{D}_{s}$ across various inter-request and inter-server distance distributions in Figure \ref{1d-grid-temp-vs-spat-deterministic}. It is evident that as $\rho \to 1,$ the ratio $\mathbb{E}[D]/\overline{D}_{s}$ converges to $1$ across different inter-server distance distributions. 

\subsection{Heterogeneous server capacities under PRGS}\label{sec:het}
We now proceed to analyze a setting where server capacity is a random variable. Assume server capacity $\mathcal C$ takes values from $\{1,2,\ldots,c\}$ with distribution $\texttt{Pr}(\mathcal C=j) = p_j, \forall j\in\{1,2,\ldots,c\}$, s.t. $\sum_{j=1}^c p_j= 1$ and $p_c > 0.$ We also assume the stability condition $\rho<\overline{\mathcal C}$ where $\overline{\mathcal C}$ is the average server capacity. Denote $H$ as the random variable associated with number of requests that traverse through a point just after a server location\footnote{An analysis for the distribution of number of requests that traverse through any random location would involve the notions of exceptional service and accessible batches.}.
\subsubsection{Distribution of $H$}
Let $V$ denote the number of new requests generated during a service period with $k_v = \texttt{Pr}(V=v), \forall v\ge0.$ According to the law of total probability, it holds that
\begin{align}
k_v = \int\limits_0^\infty \texttt{Pr}(V=v|X=\nu)f_X(\nu)= \frac{1}{v!}\int\limits_0^\infty e^{-\lambda\nu}(\lambda\nu)^vdF_X(\nu).
\end{align} 
\noindent Then the corresponding generating function $K(z)$ is denoted by
\begin{align}
K(z) = \sum\limits_{v=0}^{\infty}k_vz^v = F_X^{*}(\lambda(1-z)).
\end{align} 
We now consider an embedded Markov chain generated by $H$. Denote the corresponding transition matrix as $M.$ Then we have
{\footnotesize
\begin{equation}\label{eq:utility}
    M_{m,l}=
\begin{cases}
     \sum\limits_{i=0}^{c-m}k_iP_{i+m},& 0\le m\le c, l = 0;\\
     \sum\limits_{i=0}^{c}k_{i+l-m}p_i,& 0\le m\le l, l \ne 0;\\
     \sum\limits_{i=m-l}^{c}k_{i+l-m}p_i,& l+1\le m\le c+l, l \ne 0;\\
     0, & o.w.,
\end{cases}
\end{equation}
}	

\noindent where $P_{i} = \sum_{j=i}^{c}p_j$ and $p_0 = 0.$ Let $\pi = [\pi_j, j\ge0]$ and $N(z) = \sum_{j\ge0}\pi_jz^j$  denote the steady state distribution and its $z$-transform respectively. $\pi$ is obtained out by solving
\begin{align}
\pi_l = \sum\limits_{m=0}^{\infty}\pi_mM_{m,l}, l=0,1,\ldots.
\end{align}
Thus we have for $l\in\mathbb{N},$
\begin{align}
\pi_0 &= \sum\limits_{m=0}^{c}\pi_m\sum\limits_{i=0}^{c-m}k_iP_{i+m}, \;\pi_l = \sum\limits_{m=0}^{l}\pi_m\sum\limits_{i=0}^{c}k_{i+l-m}p_i + \sum\limits_{m=l+1}^{c+l}\pi_m\sum\limits_{i=m-l}^{c}k_{i+l-m}p_i.
\end{align} 
Multiplying by $z^l$ and summing over $l$ gives
{\footnotesize
\begin{align}
N(z) &= E_{\pi} + v_1(z) + v_2(z)\label{eq:pi_var_c}\\
E_{\pi} &= \pi_0\sum\limits_{i=0}^{c-1}k_iP_{i+1}+\sum\limits_{m=1}^{c-1}\pi_m\sum\limits_{i=m}^{c-1}k_{i-m}P_{i+1}\\
v_1(z) &= \sum\limits_{l=0}^\infty z^l \sum\limits_{m=0}^{l}\pi_m\sum\limits_{i=0}^{c}k_{i+l-m}p_i \label{eq:lhs_var_c}\\
v_2(z) &= \sum\limits_{l=0}^\infty z^l \sum\limits_{m=l+1}^{c+l}\pi_m\sum\limits_{i=m-l}^{c}k_{i+l-m}p_i.\label{eq:rhs_var_c}
\end{align} 
}

The expressions for $v_1(z)$ and $v_2(z)$ can be further simplified (see Appendix \ref{app-het}) to
{\footnotesize
\begin{align}
v_1(z) &= N(z)\bigg\{\sum\limits_{i=0}^{c}p_iz^{-i}\bigg[K(z)-\sum\limits_{j=0}^{i}k_jz^j\bigg]+\sum\limits_{i=0}^{c}k_iz^i\bigg\}\label{eq:l_var_c}\\
v_2(z) &= \bigg[\sum\limits_{m=0}^{c}z^{-m}\sum\limits_{i=m}^{c}k_{i-m}p_i\bigg\{N(z)-\sum\limits_{j=0}^{m-1}\pi_jz^j\bigg\}\bigg] - N(z)\sum\limits_{i=0}^{c}k_iz^i.\label{eq:r_var_c}
\end{align} 
}

Combining \eqref{eq:pi_var_c}, \eqref{eq:l_var_c} and \eqref{eq:r_var_c} yields
{\footnotesize
\begin{align}
N(z) = E_{\pi} &+ N(z)\bigg\{K(z)\sum\limits_{i=0}^{c}p_iz^{-i}\bigg\}-\sum\limits_{j=0}^{c-1}\pi_j\sum\limits_{m=1}^{c-j}z^{-m}\sum\limits_{i=m+j}^{c}k_{i-(m+j)}p_i.
\end{align} 
}
Thus we obtain
{\footnotesize
\begin{align}\label{eq:het_l}
N(z) = \frac{E_{\pi} -\sum\limits_{j=0}^{c-1}\pi_j\sum\limits_{m=1}^{c-j}z^{-m}\sum\limits_{i=m+j}^{c}k_{i-(m+j)}p_i}{1-K(z)\sum\limits_{i=0}^{c}p_iz^{-i}}.
\end{align} 
}
Multipying numerator and denominator by $z^c$ yields
{\footnotesize
\begin{align}\label{eq:het_l}
N(z) = \frac{z^cE_{\pi} -\sum\limits_{j=0}^{c-1}\pi_j\sum\limits_{m=1}^{c-j}z^{c-m}\sum\limits_{i=m+j}^{c}k_{i-(m+j)}p_i}{z^c-K(z)\sum\limits_{i=0}^{c}p_{c-i}z^{i}}.
\end{align} 
}
To determine $N(z)$, we need to obtain the probabilities $\pi_i,0\le i\le c-1.$ It can be shown that the denominator of  \eqref{eq:het_l} has $c-1$ zeros inside and one on the unit circle, $|z|= 1$ (See Appendix \ref{app-het-rouche}). As $N(z)$ is analytic within and on the unit circle, the numerator must vanish at these zeros, giving rise to $c$ equations in $c$ unknowns.

Let $\xi_q: 1\leq q\leq c$ be the zeros of $z^c-K(z)\sum_{i=0}^{c}p_{c-i}z^{i}$ in $\{|z|\leq 1\}$. W.l.o.g let $\xi_c = 1.$ We have the following $c-1$ equations.
{\footnotesize
\begin{align}
E_{\pi} -\sum\limits_{j=0}^{c-1}\pi_j\sum\limits_{m=1}^{c-j}\xi_q^{-m}\sum\limits_{i=m+j}^{c}k_{i-(m+j)}p_i=0, \,\, i=1,\ldots, c-1,
\end{align}
}

A $c$-th equation is provided by the normalizing condition $ \lim_{z\to 1}$ $N(z)=1$. In the particular case where all zeros have multiplicity one, it can be shown that these $c$ equations are linearly independent\footnote{For all cases evaluated across uniform, deterministic and hyperexponential distributions we found the set of $c$ equations to be linearly independent.}. Once the parameters $\{\pi_i, 0\le i \le c-1\}$ are known, $\mathbb{E}[H]$ can be expressed as
\begin{align}
\mathbb{E}[H] = \overline{H} = \lim_{z\to 1}N^{\prime}(z).
\end{align}
\subsubsection{Expected Request Distance}
To evaluate the expected request distance we adopt arguments from \cite{bailey54}. Consider any interval of length $\nu$ between two consecutive servers. There are on average  $\overline{H}$ requests at the beginning of the interval , each of which must travel $\nu$ distance. New users are spread randomly over the interval and there are on an average $\lambda\nu$ new users. The request made by each new user must travel on average $\nu/2.$ Thus we have
\begin{align}\label{eq:erd_var_c}
\mathbb{E}[D] &= \frac{1}{\rho}\int_0^\infty(\overline{H}\nu+\frac{1}{2}\lambda\nu^2)dF_X(\nu) = \frac{1}{\rho}\bigg[\frac{\overline{H}}{\mu}+\frac{\lambda}{2}\bigg(\sigma_X^2+\frac{1}{\mu^2}\bigg)\bigg].
\end{align}


\subsection{Uncapacitated request allocation}\label{sub:uncap2}
An interesting special case of the unidirectional general matching is the uncapacitated scenario. Consider the case where servers do not have any capacity constraints, i.e. $c=\infty.$ In such a case, all users are assigned to the nearest server to their right.\\

\noindent{\bf GRPS:} When $c\to\infty$ and given $0<r_0<1,$  $r_0 = F^*_Y(\mu-\mu r_0^c) = F^*_Y(\mu).$ Setting $\omega = 1/F^*_Y(\mu) = 1/r_0$ in \eqref{eq:probs_C} and simplifying yields
\begin{align}\label{eq:probs_C2}
C \to 0, \;\texttt{as}\; c\to\infty,\implies \mathbb{E}[D] \to \frac{1}{\mu} \;\texttt{as}\; c\to\infty.
\end{align}
\noindent{\bf PRGS:} Under PRGS, when $c\to\infty$ there exists no request allocated to a server other than the nearest server to its right. Again using Bailey's method as in \cite{bailey54} and setting $\overline{H} = 0$ in \eqref{eq:erd_var_c} we get
\begin{align}\label{eq:probs_prgs}
\mathbb{E}[D] \to \frac{\mu}{2}\bigg(\sigma_X^2+\frac{1}{\mu^2}\bigg) \;\texttt{as}\; c\to\infty.
\end{align}

\subsection{Cost models}\label{sub:cost_model}
\np{
Consider the following generalization of the service network. We define cost of an allocation as the communication cost associated with an allocated request-server pair. Consider communication cost as a function ${\mathcal T}$ of the request distances. Then the expected communication cost across the service network is given as
\begin{align}
\overline{T} = E[\texttt{cost}] = \int\limits_{d=0}^{\infty}{\mathcal T}(d) dW(d),
\end{align}
\noindent where $W$ is the request distance distribution. One such cost model widely used in wireless ad hoc networks is \cite{Doshi02}
\begin{align}
{\mathcal T}(d) = t_0d^{\beta},
\end{align}
}
\noindent \np{where $\beta$ is the path loss exponent typically $2\le \beta\le 4$ and $t_0$ is a constant. Below we derive  the expected communication cost for the scenario when $c=1.$}
\subsubsection{GRPS with $c=1$} \np{In this case the service network directly maps to a temporal G/M/1 queue. Thus $W$ can be expressed as the sojourn time distribution of the corresponding G/M/1 queue. Hence $W\sim\text{Expo}(\mu(1-r_0))$ with $r_0$ as defined in Section \ref{sub:grps1}. We have
\begin{align}
\overline{T} &= \int\limits_{d=0}^{\infty}t_0d^\beta dW(d)= \frac{t_0}{\mu^\beta(1-r_0)^\beta}\Gamma(\beta+1),
\end{align}
}
\noindent \np{where $\Gamma(x) = \int_{0}^{\infty}y^{x-1} e^{-y}dy$ is the gamma function.}
\subsubsection{PRGS with $c=1$} \np{In this case, the service network can be modeled as a temporal M/G/1 queue with first customer having exceptional service \cite{Welch64}. Denote $W^*(s)$ as the LS transform of  $W.$ Using results from \cite{Welch64}
\begin{align}
W^*(s) = \frac{(1-\rho)\left\{\lambda\bigg[F_Z^*(s)-F_X^*(s)\bigg]-sF_Z^*(s)\right\}}{(1-\rho+\rho_Z)\bigg[\lambda-s-\lambda F_X^*(s)\bigg]}.
\end{align}
When $\beta$ is an integer,
\begin{align}
\overline{T} &= t_0(-1)^{\beta}\frac{d^{(\beta)}}{ds}W^*(s)|_{s=0},
\end{align}
}
\np{
\subsection{Extension to two resources}\label{sec:exttwo}
\begin{figure}[]
\centering
\begin{minipage}{0.35\textwidth}
\includegraphics[width=0.9\textwidth]{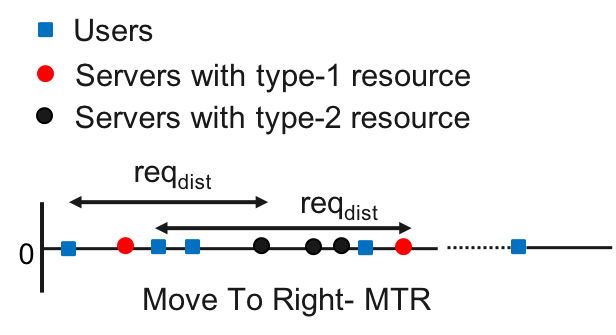}
\subcaption{}
\end{minipage}
\begin{minipage}{0.35\textwidth}
\includegraphics[width=0.9\textwidth]{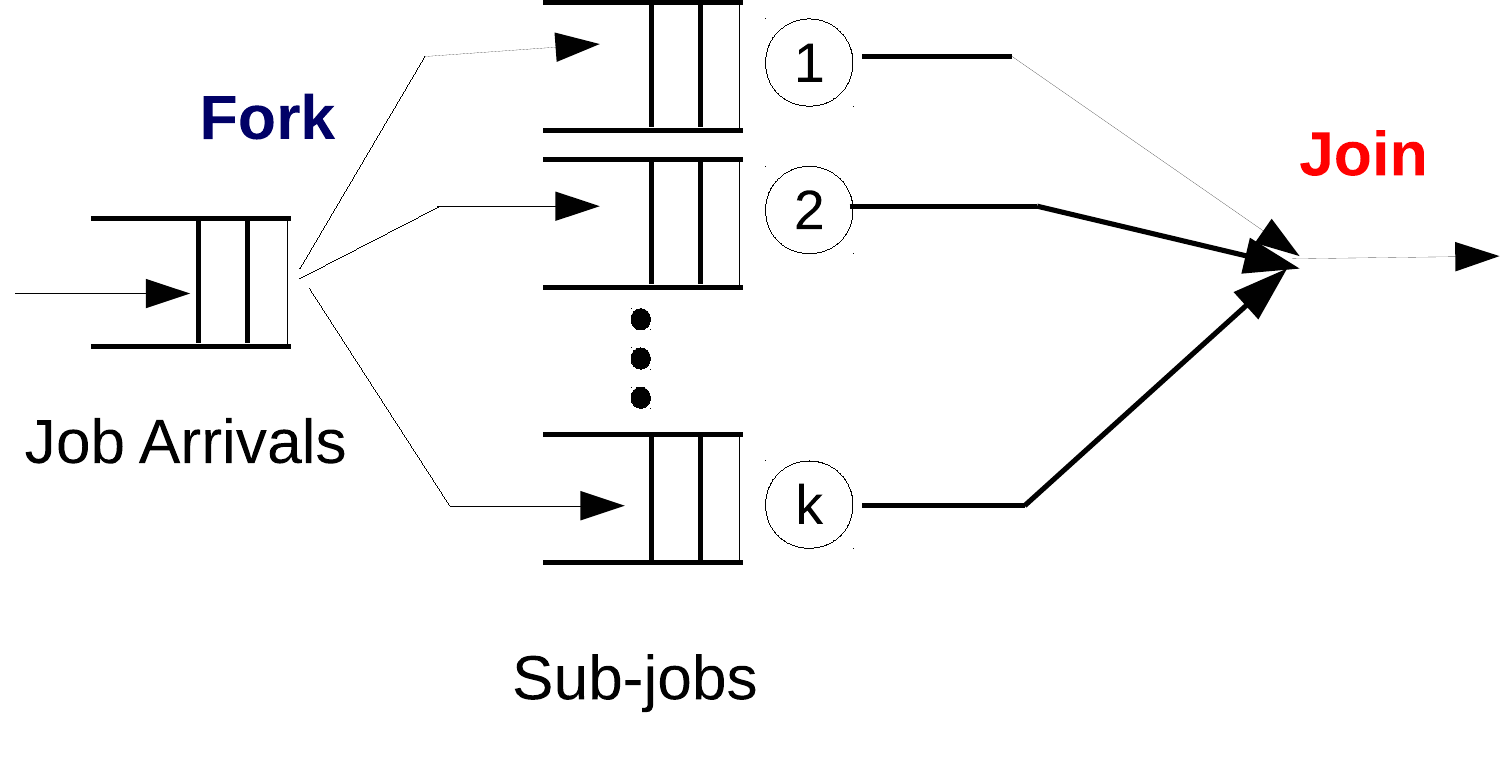}
\subcaption{}
\end{minipage}
\caption{Two resource scenario with $c=1$ (a) Depiction of request distances and (b) Mapping to Fork-join queues.}
\label{fig:two_r}
\vspace{-0.15in}
\end{figure}
Now consider the following scenario where each user requests two resources which reside on different servers as shown in Figure \ref{fig:two_r}(a). Let the corresponding servers be distributed according to a Poisson process with densities $\mu_1$ and $\mu_2$. Let the users be distributed according to a Poisson process. The service network, in this case, can be modeled as a fork-join queueing system as shown in Figure \ref{fig:two_r}(b) \cite{nelson88}. In such a queue, each incoming job is split into two sub-jobs each of which is served on one of the two servers. After service, each sub-job waits until the other sub-job has been processed. They then merge and leave the system. In the service network as well, each request forks two sub-requests one for each resource type. A request is said to be completed only if it has retrieved both the resources, thus mapping it to a fork-join queue. We define the {\it overall request distance} to be the maximum value among the request distances across all resource types and denote it as the random variable $D_{max}$.
\subsubsection{Identical service rates ($\pmb{\mu_1=\mu_2 = \mu}$ and $\pmb{c = 1}$)}
The approximated expected request distance for this scenario is obtained from the expression for the expected sojourn time of a fork join queue with homogeneous servers as \cite{nelson88}:
\begin{align}
\mathbb{E}[D_{max}] = \frac{12\mu-\lambda}{8\mu(\mu-\lambda)}, \label{eq:two-res}
\end{align}
\noindent Note that, the corresponding expected request distance in case of single resource is given by Equation \eqref{eq:c12} $\mathbb{E}[D] = 1/(\mu-\lambda)$. Clearly, 
\begin{align}
\mathbb{E}[D_{max}] = \frac{12\mu-\lambda}{8\mu(\mu-\lambda)} = \left[1.5- 0.125 \rho\right] \frac{1}{\mu-\lambda} > \frac{1}{\mu-\lambda} =  \mathbb{E}[D],
\label{eq:two-res-comp1}
\end{align}
\noindent Thus we have $\mathbb{E}[D_{max}] > \mathbb{E}[D].$
}

%% file: 08-opt-assignment.tex

Both UGS and MTR minimize expected request distance among all unidirectional policies. In this section we formulate the bi-directional allocation policy that minimizes expected request distance. Let $\eta: R \to S$ be any mapping of users to servers. Our objective is to find a mapping $\eta^*: R \to S$, that satisfies
\begin{eqnarray} 
\nonumber& & \eta^* = \arg\min_{\eta} \sum_{i\in R} d_{\mathcal{L}} (r_i,s_{\eta(i)})\\
& s.t. & \sum_{i\in R} \pmb{1}_{\eta(i)=j} \leq c, \forall j\in S \label{eq:assg}
\end{eqnarray}
W.l.o.g, let $r_1\le r_2\le\cdots \le r_i\le\cdots \le r_{|R|}$ be locations of requests and $s_1\le s_2\le\cdots \le s_i\le\cdots\le s_{|S|}$ be locations of servers. We first focus on the case when $c=1$. We consider the following two scenarios.\\


\noindent {\bf Case 1: \pmb{$|R|=|S|$} }\\
When $|R|=|S|$, an optimal allocation strategy is given by the following theorem \cite{Bukac18}.
\begin{theorem}
When $|R|=|S|$, an optimal assignment is obtained by the policy: $\eta^*(i) = i, \;\forall i \in \{1,\cdots,|R|\}$ i.e. allocating the $i^{th}$ request to the $i^{th}$ server and the average request distance is given by
\begin{align}
\mathbb{E}[D] = \frac{1}{|R|}\sum\limits_{i=1}^{|R|}|s(i)-r(i)|. \label{eq:c1}
\end{align}
\label{th:equal_r_s}
\end{theorem}

\noindent {\bf Case 2: \pmb{$|R|<|S|$} }
This is the case where there are fewer requesters than servers. In this case, a Dynamic Programming ({\it DP}) based  algorithm (Algorithm \ref{algo:opt-dps}) obtains the optimal assignment. 

Let $C[i,j]$ denote the optimal cost (i.e., sum of distances) of assigning the first $i$ requests (counting from the left) located at $r_1\leq r_2\leq \ldots\leq r_i$ to the first $j$ servers (also counting from the left) located at  $s_1\leq s_2\leq \ldots\leq s_j$. If $j==i$, the optimal assignment is trivial due to Theorem \ref{th:equal_r_s} and $C[i,i]$ is computed easily for all $i \leq |R|$ by summing pairwise distances $d[1,1],d[2,2],\ldots,d[i,i]$ (Lines 6--7).
For the base case, $i=1, j>1$, only the first user needs to be assigned to its nearest server (Lines 9--16).
  For the general dynamic programming step, consider $j > i$. Then $C[i,j]$ can be expressed in terms of the costs of two subproblems, i.e., $C[i-1,j-1]$ and $C[i,j-1]$ (Lines 19--24). In the optimal solution, two cases are possible: either request $i$ is assigned to server $j$, or the latter is left unallocated. The former case occurs if the first $i-1$ requests are assigned to the first $j-1$ servers at cost $C[i-1,j-1]$, and the latter case occurs when the first $i$ requests are assigned to the first $j-1$ servers at cost $C[i,j-1]$. This is a consequence of the no-crossing lemma (Lemma \ref{lem:nocross}). The optimal $C[i,j]$ is chosen depending on these two costs and the current distance $d[i,j]$.

\begin{lemma}\label{lem:nocross}
In an optimal solution, $\eta^*,$ to the problem of matching users at $r_1\leq r_2\leq \ldots\leq r_{|R|}$ to servers at $s_1\leq s_2\leq \ldots\leq s_{|S|}$, where $|S| \geq |R|$, there do not exist indices $i,j$ such that $\eta^*(i)>\eta^*(i')$ when $i' > i$.
\begin{proof}
See Appendix \ref{app-lemma}.
\end{proof}
\end{lemma}


The dynamic programming algorithm fills cells in an $|R|\times |S|$ matrix $C$ whose origin is in the north-west corner. The lower triangular portion of this matrix is invalid since $|R| \leq |S|$. The base cases populate the diagonal and the northernmost row, and in the general DP step, the value of a cell depends on the previously computed values in the cells located to its immediate west and diagonally north-west. As an optimization, for a fixed $i$, the $j$-th loop index needs to run only from $i+1$ through $i+|S|-|R|$ (Lines 11 and 18) instead of from $i+1$ through $|S|$. This is because the first request has to be assigned to a server $s_j$ with $j \leq |S|-|R|+1$ so that the rest of the $|R|-1$ requests have a chance of being placed on unique servers\footnote{Note that in this exposition, we consider server capacity $c=1$. If $c>1$, we simply add $c$ servers at each prescribed server location, and requests will still be placed on unique servers.}. The optimal average request distance is given by $C[|R|,|S|]$.

The time complexity of the main DP step is $O(|R| \times (|S|-|R|+1))$. Note that this assumes that the pairwise distance matrix $d$ of dimension $|R|\times |S|$ has been precomputed. The optimization applied above can be similarly applied to this computation and hence the overall time complexity of Algorithm \ref{algo:opt-dps} is $O(|R| \times (|S|-|R|+1))$. Therefore, if $|S|=O(|R|)$, the worst case time complexity is quadratic in $|R|$. However, if $|S|-|R|$ grows only sub-linearly with $|R|$, the time complexity is sub-quadratic in $|R|$.

Note that retrieving the optimal assignment requires more book-keeping. An $|R|\times |S|$ matrix $A$ stores key intermediate steps in the assignment as the DP algorithm progresses (Lines 8, 16, 21, 24). The optimal assignment vector $\pi$ can be retrieved from matrix $A$ using procedure {\sc ReadOptAssignment}.

\begin{algorithm}[ht]
\begin{algorithmic}[1]
\State \textbf{Input}: $r_1\le \cdots \le r_{|R|}$; \;$s_1\le\cdots\le s_{|S|}$
\State \textbf{Output}: The optimal assignment $\pi$

\Procedure{OptDP}{$r,s$}
	\State $d_{|R|\times|S|} = \Call{ComputePairwiseDistances}{r,s}$
	\State $C = \{ \infty \}_{|R|\times|S|}$
	\For {$i = 1,\cdots,|R|$}
		\State $C[i,i] = \Call{TrivialAssignment}{i,d}$
	\EndFor
	\State $A[|R|,|R|] = |R|$
	\State $nearest = 0$
	\State $nearestcost = C[1,1]$
	\For {$j = 2, \cdots, |S|-|R|+1$}
		\If {$d[1,j] < nearestcost$}
			\State $nearestcost = d[1,j]$
			\State $nearest = j$
		\EndIf
		\State $C[1,j] = nearestcost$
		\State $A[1,j] = nearest$
	\EndFor
	\For {$i = 2,\cdots,|R|$}
		\For {$j = i+1,\cdots,i+|S|-|R|$}
			\If {$C[i,j-1] < d[i,j]+C[i-1,j-1]$}
				\State $C[i,j] = C[i,j-1]$
				\State $A[i,j] = A[i,j-1]$
			\Else
				\State $C[i,j] = d[i,j]+C[i-1,j-1]$
				\State $A[i,j] = j$
			\EndIf
		\EndFor
	\EndFor
	\State \Return $\Call{ReadOptAssignment}{A}$
\EndProcedure

\Procedure{TrivialAssignment}{$n,d$}
  \State $Cost = 0$
  \For{$i = 1,\cdots,n$}
    \State $Cost = Cost + d[i,i]$
  \EndFor
  \State \Return $Cost$
\EndProcedure

\Procedure{ReadOptAssignment}{$A$}
  \State ${|R|,|S|} = \Call{Dimensions}{A}$
  \State $s = |S|$
  \For{$i = |R|, \cdots, 1$}
    \State $\pi[i] = A[i,s]$
    \State $s = A[i,s]-1$
  \EndFor
  \State \Return $\pi$
\EndProcedure

\end{algorithmic}
\caption{Optimal Assignment by Dynamic Programming}
\label{algo:opt-dps}
\vspace{-0.1cm}
\end{algorithm}
\begin{figure}[]
\centering
\includegraphics[width=1.0\columnwidth]{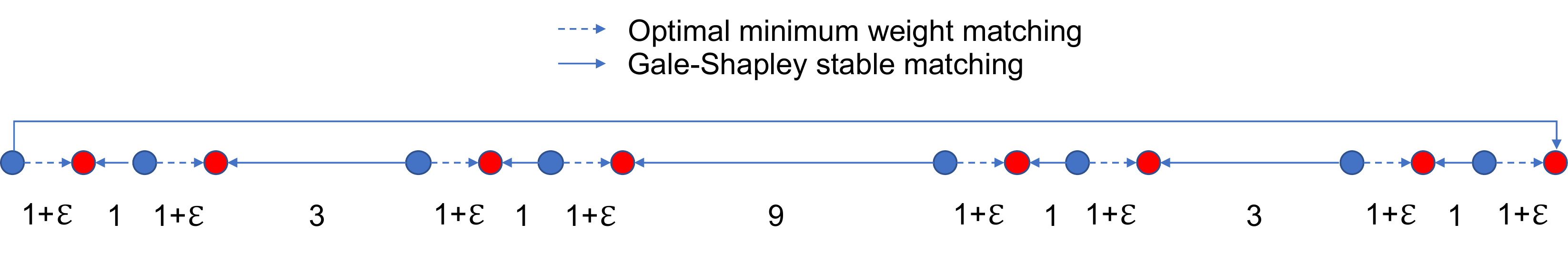}
\caption{\label{fig:opt-vs-gs}Worst case scenario for Gale-Shapley.}
\end{figure}

Another bidirectional assignment scheme is the Gale-Shapley algorithm~\cite{Gale62}, which produces stable assignments, though in the worst case it can yield an assignment that is $O(|R|^{\ln{3/2}})\approx O(|R|^{0.58})$ times costlier than the optimal assignment yielded by Algorithm \ref{algo:opt-dps}, where $|R|$ is the number of users~\cite{RT81}. The worst case scenario is illustrated in Figure \ref{fig:opt-vs-gs}, with $|R| = 2^{t-1}$, where $t$ is the number of clusters of users and servers; and the largest distance between adjacent points is $3^{t-2}$. However at low/moderate loads for the cases evaluated in Section \ref{sec:perfcomp}, we find its performance to be not much worse than optimal.

%% file: 09-perf-comp.tex

In this section, we examine the effect of various system parameters on expected request distance under MTR policy. We also compare the performance of various greedy allocation strategies along with the unidirectional policies to the optimal strategy. 
\subsection{Experimental setup}
In our experiments, we consider a mean requester rate $\lambda \in (0,1).$ We consider various inter-server distance distributions with density one. In particular, (i) for exponential distributions, the density is set to $\mu=1$; (ii) for deterministic distributions, we assign parameter $d_0 = 1.$ (iii) for second order hyper-exponential distribution ($H_2$), denote $p_1$ and $p_2$ as the phase probabilities. Let $\mu_1$ and $\mu_2$ be corresponding phase rates. We assume $p_1/\mu_1  = p_2/\mu_2$. We express $H_2$ parameters in terms of the squared coefficient of variation, $c_v^2$, and mean inter-server distance, $\alpha_X$, i.e. we set $p_1 = (1/2)\big(1+\sqrt{(c_v^2-1)/(c_v^2+1)}\big), p_2 = 1-p_1, \mu_1 = 2p_1/\alpha_X$ and $\mu_2 = 2p_2/\alpha_X.$ Unless specified, for $H_2$ we take $c_v^2 = 4$ with $c=2.$ Also if not specified, users are distributed according to a Poisson process and servers a according to a renewal process.

We consider a collection of $10^5$ users and $10^5$ servers, i.e. $|R|=|S| = 10^5.$ We assign users to servers according to MTR. Let $R_M\subseteq R$ be the set of users allocated under MTR. Clearly $|R_M| \le |R|.$ We then run optimal and other greedy policies on the set $R_M$ and $S.$ For each of the  experiments, the expected request distance for the corresponding policy is averaged over $50$ trials. 
\subsection{Sensitivity analysis}
\begin{figure}[]
\centering
\begin{minipage}{.32\textwidth}
\centering
\includegraphics[width=1\linewidth]{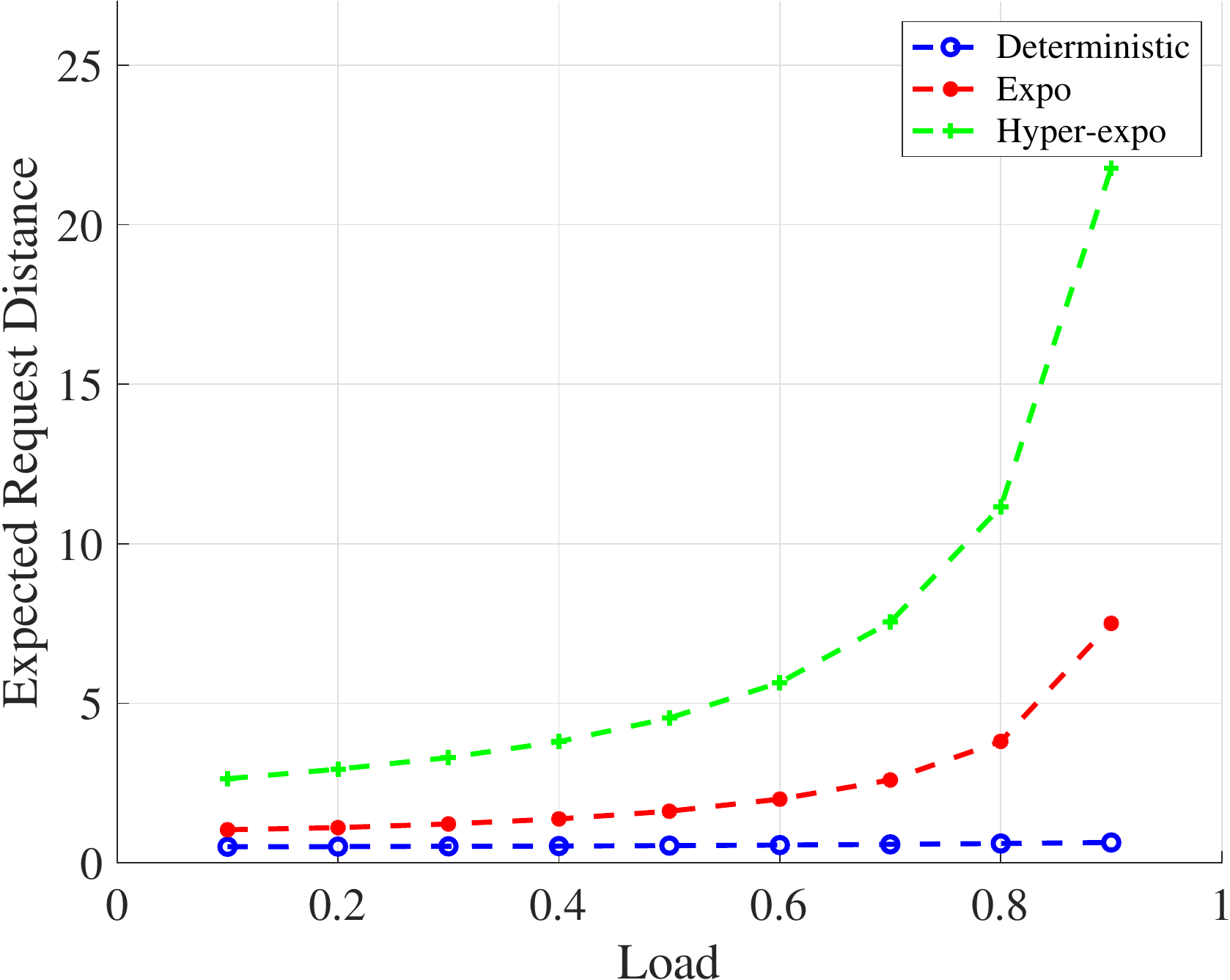}
\subcaption{}
\end{minipage}\hfill
\begin{minipage}{.32\textwidth}
\centering
\includegraphics[width=1\linewidth]{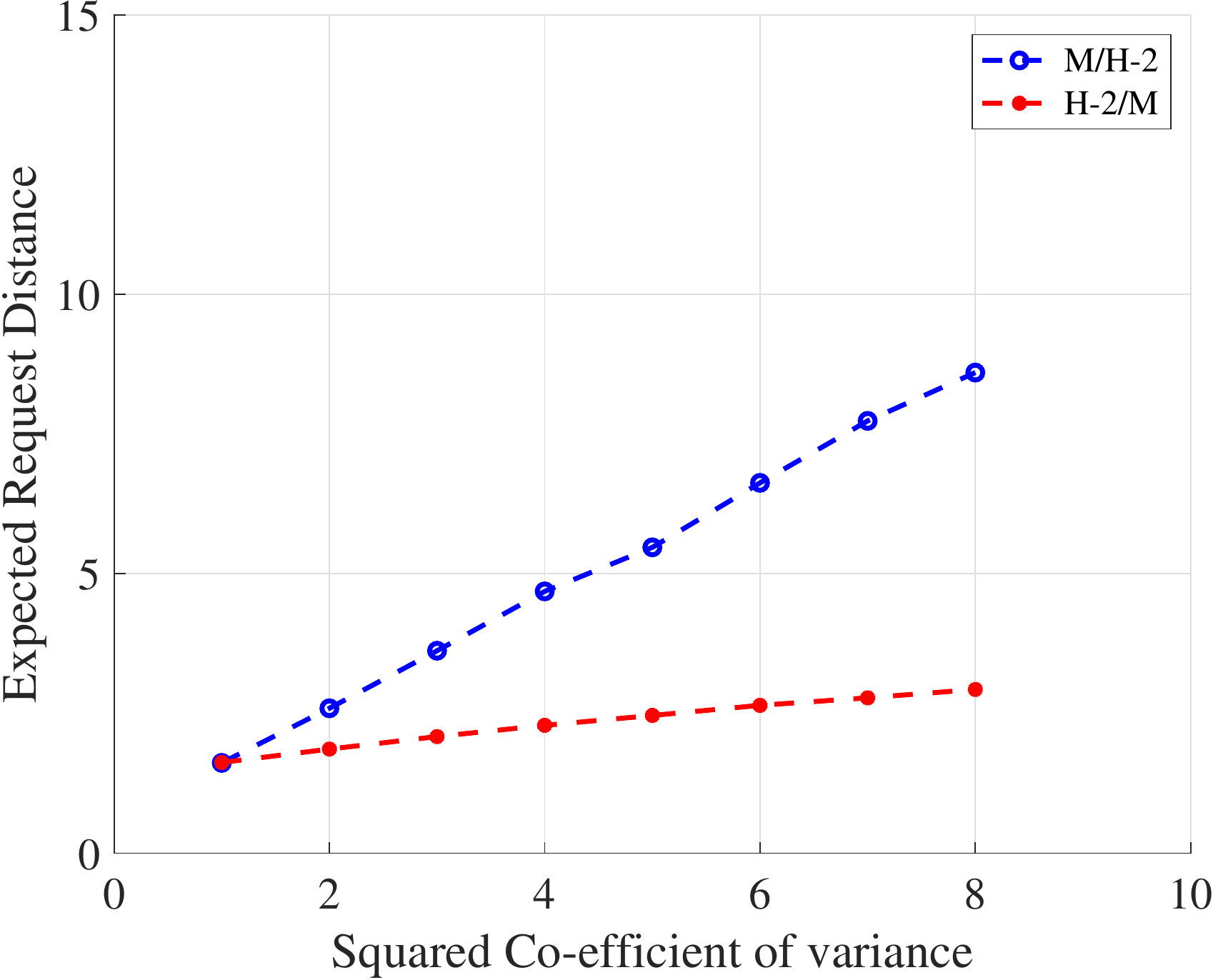}
\subcaption{}
\end{minipage}\hfill
\begin{minipage}{.32\textwidth}
\centering
\includegraphics[width=1\linewidth]{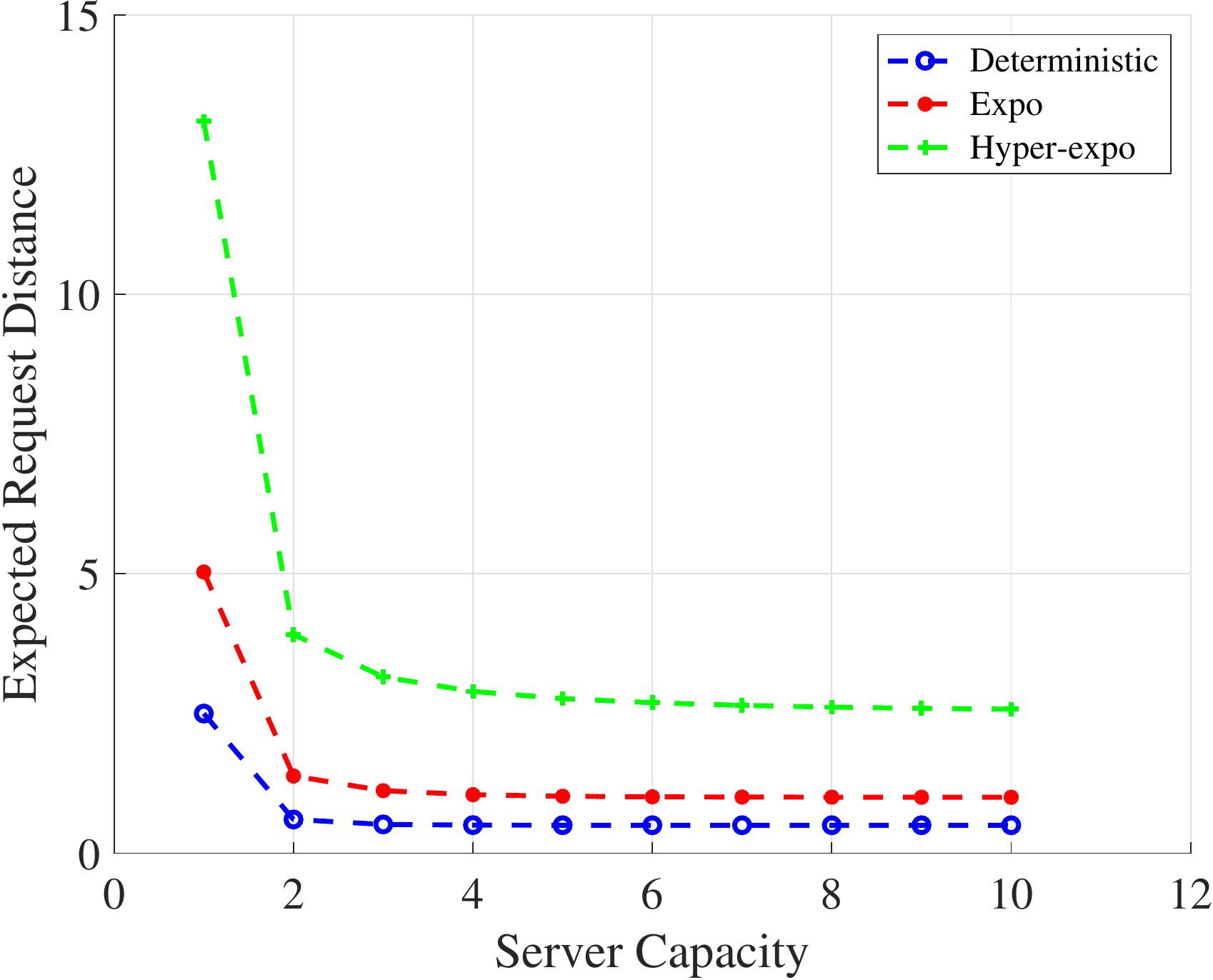}
\subcaption{}
\end{minipage}
\begin{minipage}{0.32\textwidth}
\includegraphics[width=1\textwidth]{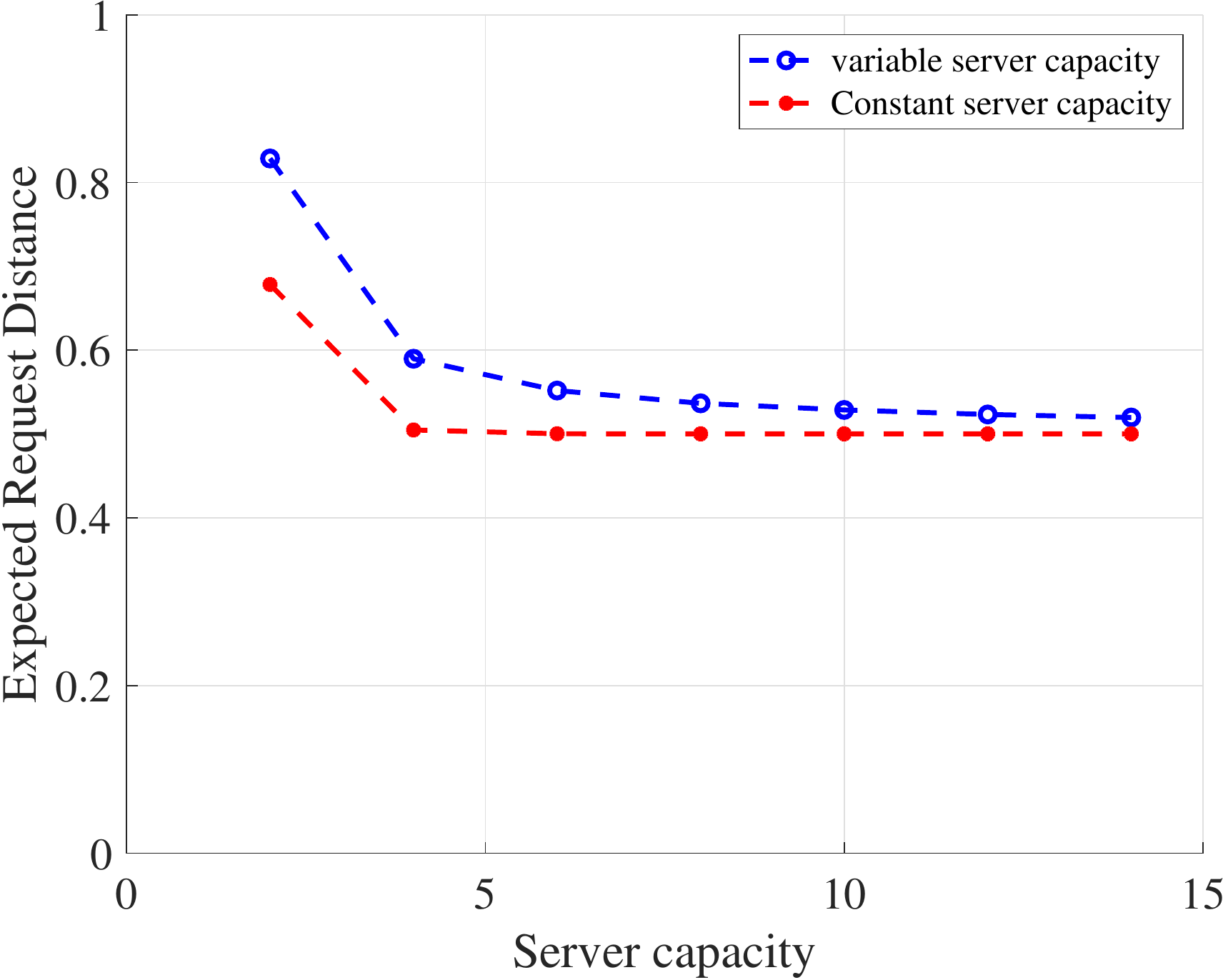}
\subcaption{}
\end{minipage}
\begin{minipage}{0.32\textwidth}
\includegraphics[width=1\textwidth]{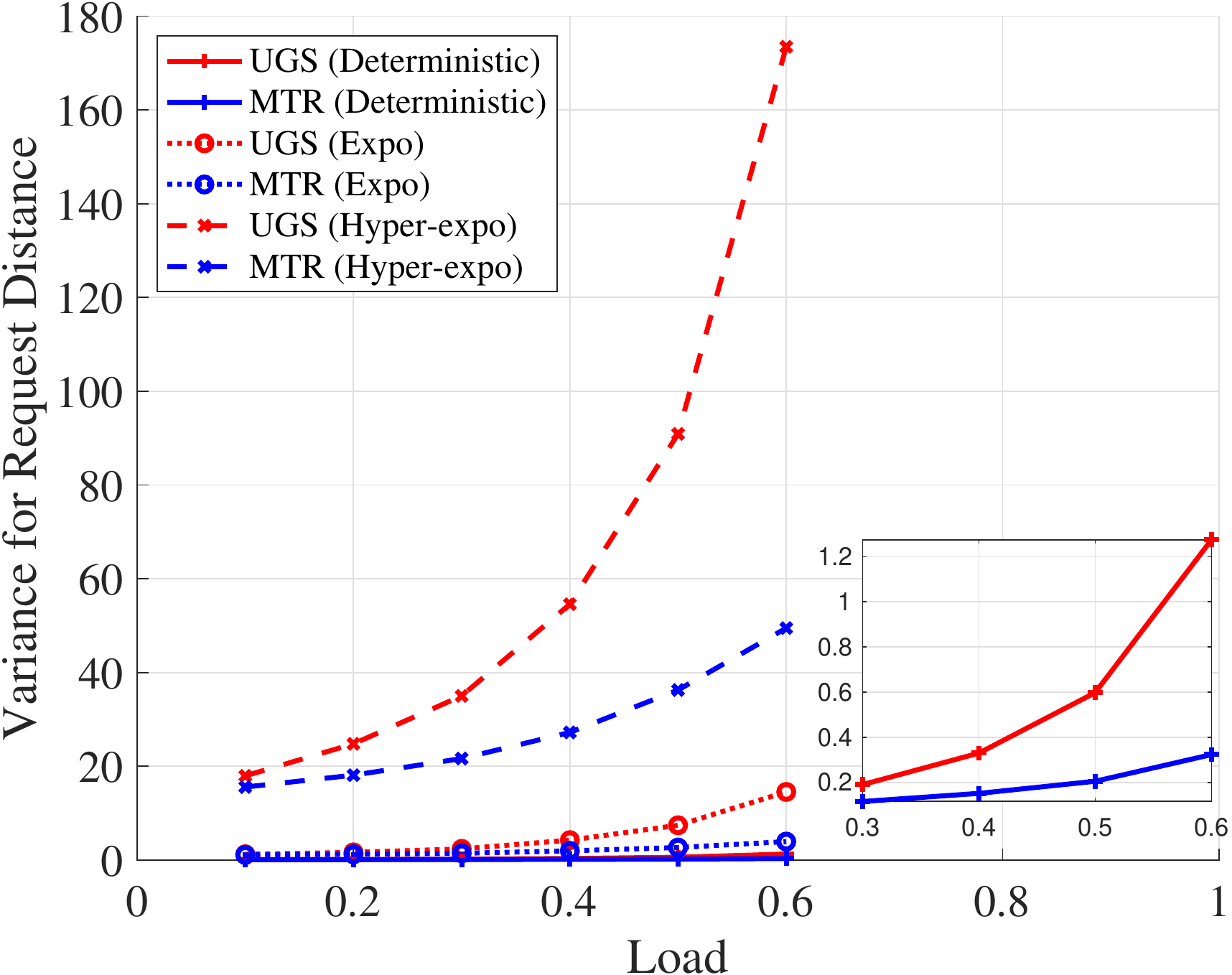}
\subcaption{}
\end{minipage}
\begin{minipage}{0.32\textwidth}
\includegraphics[width=1\textwidth]{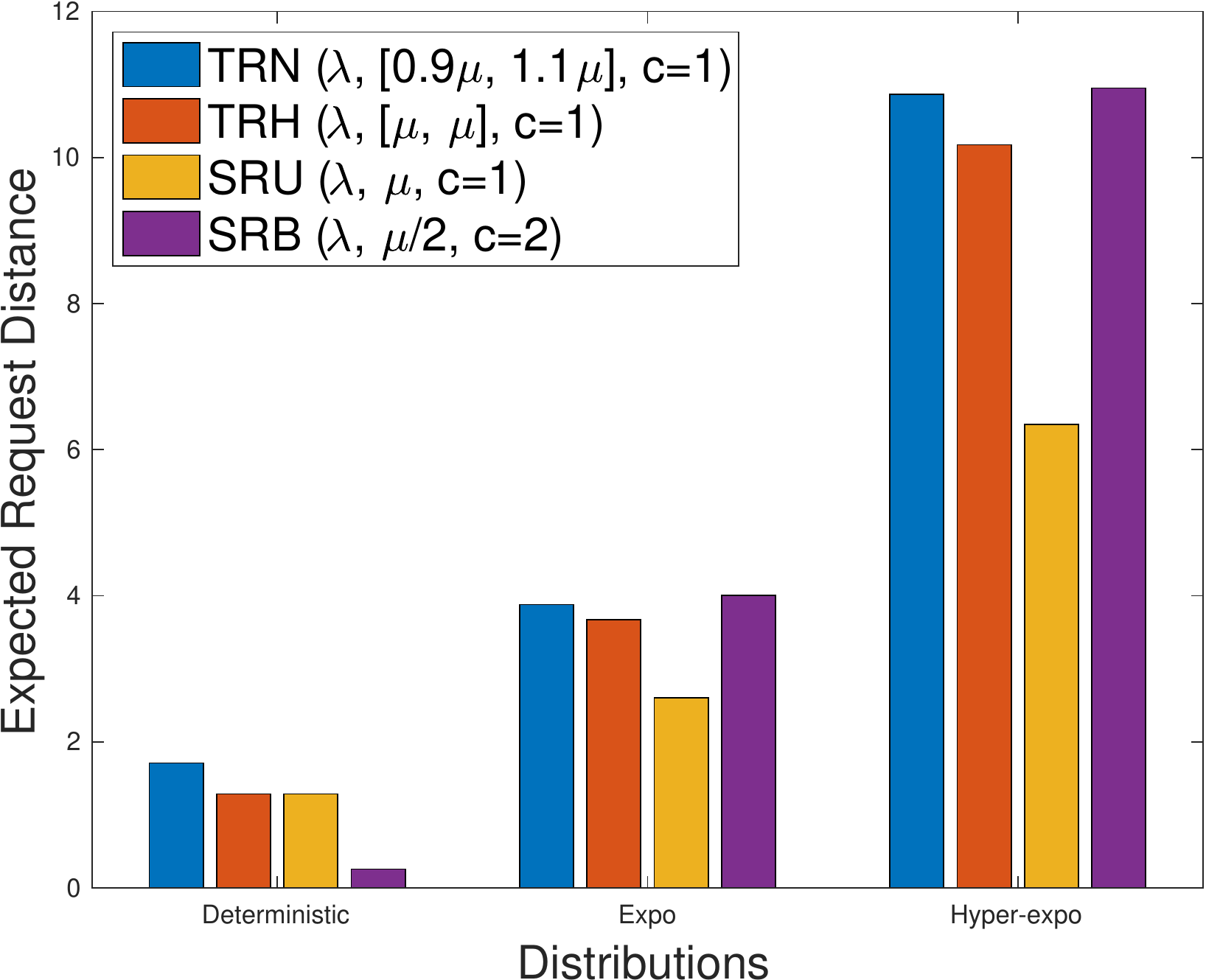}
\subcaption{}
\end{minipage}
\caption{Sensitivity analysis of MTR/UGS policy. (a) Effect of load on expected request distance with $\pmb{c=2}$. (b) Effect of squared coefficient of variation on expected request distance with $\pmb{\lambda = \mu =1}$ and $\pmb{c=2}$. (c) Effect of server capacity on expected request distance with $\pmb{\rho = 0.8}$. Effect of variability in server capacity on expected request distance for (d) Deterministic distribution with $\pmb{\rho = 0.8}$ \np{(e) Effect of load on variance of request distance with $\pmb{c=2}$ across MTR and UGS. (f) Comparison of expected request distance under Two Resource Non-homogeneous (TRN), Two Resource Homogeneous (TRH), Single Resource Unit-service (SRU)  and Single Resource Bulk-service (SRB) scenario across various server distributions with $\pmb{\lambda = 0.6, \mu = 1}$.}}
\label{fig:sensitivity}
\end{figure} 
\subsubsection{Expected request distance vs. load}
We first study the effect of load ($= \lambda/c\mu$) on $\mathbb{E}[D]$ as shown in Figure \ref{fig:sensitivity}(a). Clearly $\mathbb{E}[D]$ increases as a function of load. Note that $H_2$ distribution exhibits the largest expected request distance and the deterministic distribution, the smallest because the servers are evenly spaced. While for $H_2,$  $c_v^2$ is larger than for the exponential distribution. Consequently servers are clustered, which increases $\mathbb{E}[D].$
\subsubsection{Expected request distance vs. squared co-efficient of variation}
We now examine how $c_v^2$ affects $\mathbb{E}[D]$ when $\rho$ is fixed. We compare  two systems: a general request  with Poisson distributed servers ($H_2$/M) and a  Poisson request with general distributed servers (M/$H_2$) where the general distribution is a $H_2$ distribution with the same set of parameters, i.e. we fix $\lambda = \mu =1$ with $c=2$. The results are shown in Figure \ref{fig:sensitivity}(b). Note that, when $c_v^2=1$ $H_2$ is an exponential distribution and both $H_2$/M and  M/$H_2$ are identical M/M/1 systems. As discussed in the previous graph, performance of both systems decreases with increase in $c_v^2$ due to increase in the variability of user and server placements. However, from Figure \ref{fig:sensitivity}(b) it is clear that performance is more sensitive to  server placement as compared to the corresponding user placement.
\subsubsection{Expected request distance vs. server capacity}
We now focus on how server capacity affects $\mathbb{E}[D]$ as shown in Figure \ref{fig:sensitivity}(c). We fix $\rho = 0.8$. With an increase in $c$, while keeping $\rho$ fixed, $\mathbb{E}[D]$ decreases. This is because queuing delay decreases. Note that $\mathbb{E}[D]$ gradually converge to a value with increase in server capacity. Theoretically, this can be explained by our discussion on uncapacitated allocation in Section \ref{sub:uncap2}. As $c\to\infty$ the contribution of queuing delay to $\mathbb{E}[D]$ vanishes and  $\mathbb{E}[D]$ becomes insensitive to $c.$
\subsubsection{Expected request distance vs. capacity moments}
We investigate the heterogeneous capacity scenario as discussed in Section \ref{sec:het}. Consider the plot shown in Figure \ref{fig:sensitivity}(d). We fix $\rho = 0.8$. For the variable server capacity curve we choose a value for server capacity for each server uniformly at random from the set $\{1,2,\ldots,2c\}.$ For the constant server capacity curve we deterministically assign server capacity $c$ to each server. We observe better performance for constant server capacity curve at lower values of $c$ under Deterministic distribution. Variability in constant server case is zero, thus explaining its better performance. Both the curves exhibit similar performance under $H_2$ distribution as well.
\np{
\subsubsection{Variance vs. load}
We now study the effect of load on variance of request distance as shown in Figure \ref{fig:sensitivity}(e). Clearly variance increases  as a function of load. Also note that UGS has a higher variance as compared to MTR across all values of load and across various inter-server distance distributions. Provable results exist (from queueing theory) that among all service disciplines  the variance of the request distance (or sojourn time in queueing terminology) is minimized under MTR (a FCFS based policy) for Poisson request arrivals and exponential inter-server distances (or service times) \cite{Kingman62}. However, these results do not generalize to other inter-server distance distributions in an exceptional service accessible batch queueing discipline. Our simulation based results in Figure \ref{fig:sensitivity}(e) thus bolster our observation in Remark \ref{remark:variance} mentioned in Section \ref{sec:queue}. Again, a deterministic equidistant placement of servers produce the least variance for request distance among all other placements.
}
\np{
\subsubsection{Comparison of two resource and single resource policies} 
We compare the performance of MTR under various two resource (TR) and single resource (SR) settings as shown in Figure \ref{fig:sensitivity}(f). For a two resource setting, denote $[\mu_1, \mu_2]$ as the server densities associated with resource types $1$ and $2$ respectively as described in Section \ref{sec:exttwo}. Denote $c$ as the server capacity associated with each resource type. We define  a Two Resource Homogeneous (TRH) system to be a two resource setting with $\mu_1 = \mu_2 = \mu$. We define  a Two Resource Non-homogeneous (TRN) system to be a two resource setting with $\mu_1 \ne \mu_2$.  For simulation purpose, we chose $\mu_1 = \mu+\epsilon$ and  $\mu_2 = \mu - \epsilon$ such that the effective server density remains $\mu.$ We also choose $c = 1.$ A Single Resource Unit-service (SRU) system is a single resource system with server density $\mu$ and $c = 1.$ A Single Resource Bulk-service (SRB) system is also a single resource system with server density $\mu/2$ and $c = 2.$ Note that the request density and effective server densities ($c\mu$) are same in all settings. From Figure \ref{fig:sensitivity}(f), it is clear that TRH performs better than TRN across all server distributions. This advocates for maintaining similar densities for each resource type in a two resource system. As expected, a deterministic equidistant placement of servers produce the least expected request distance for each system among all other choice of placements. SRB in deterministic server placement scenario performs the best among all other settings. However, it does not perform well with other server distributions. Also, note that, TRH has a higher expected request distance as compared to SRU across all server distributions. Thus Equation \eqref{eq:two-res-comp1} in Section \ref{sec:exttwo} holds true even under non-markovian setting.}
\subsection{Comparison of different allocation policies}
\begin{figure}
\centering
\begin{minipage}{0.25\textwidth}
\includegraphics[width=0.9\textwidth]{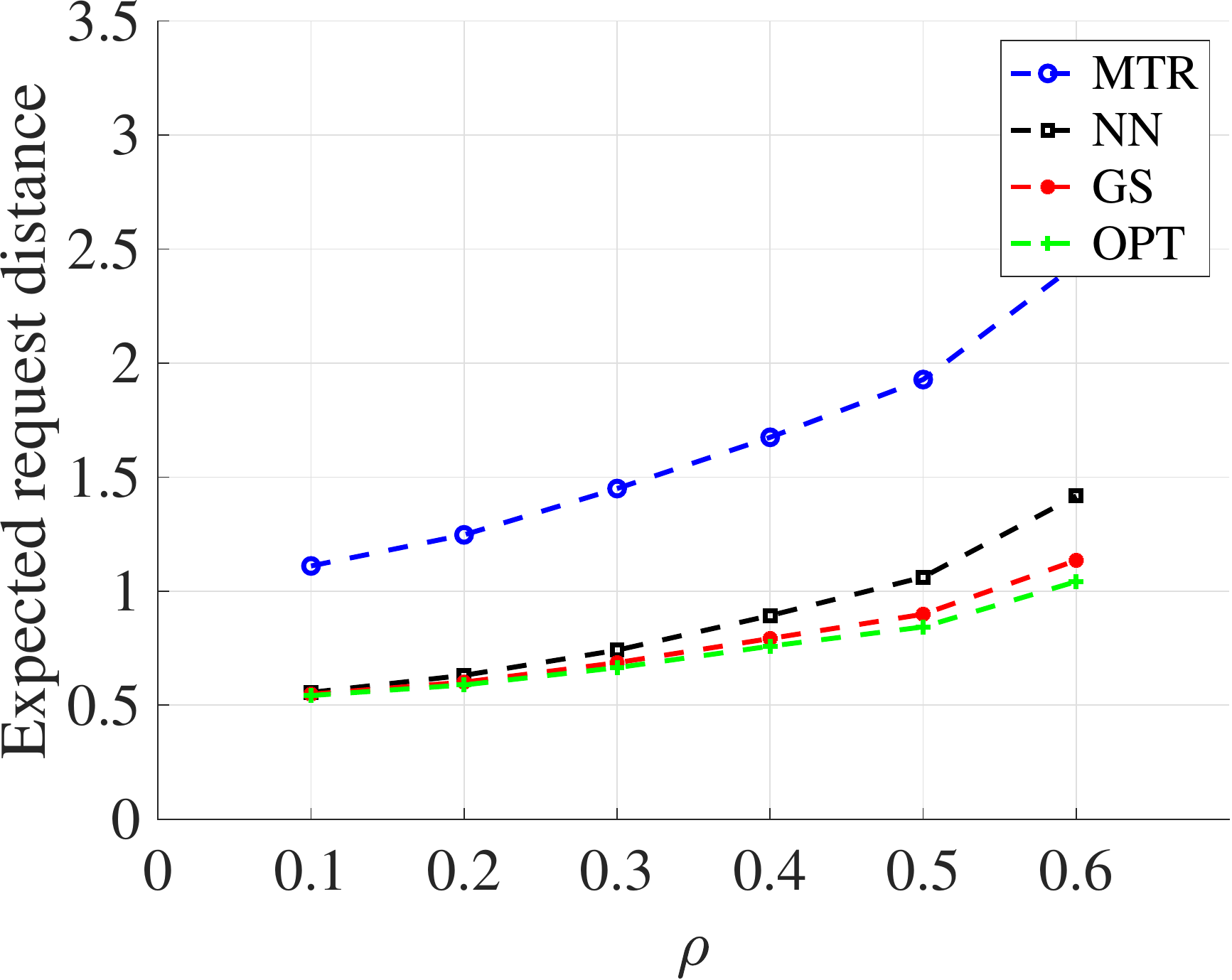}
\subcaption{}
\end{minipage}\hfill
\begin{minipage}{0.25\textwidth}
\includegraphics[width=0.9\textwidth]{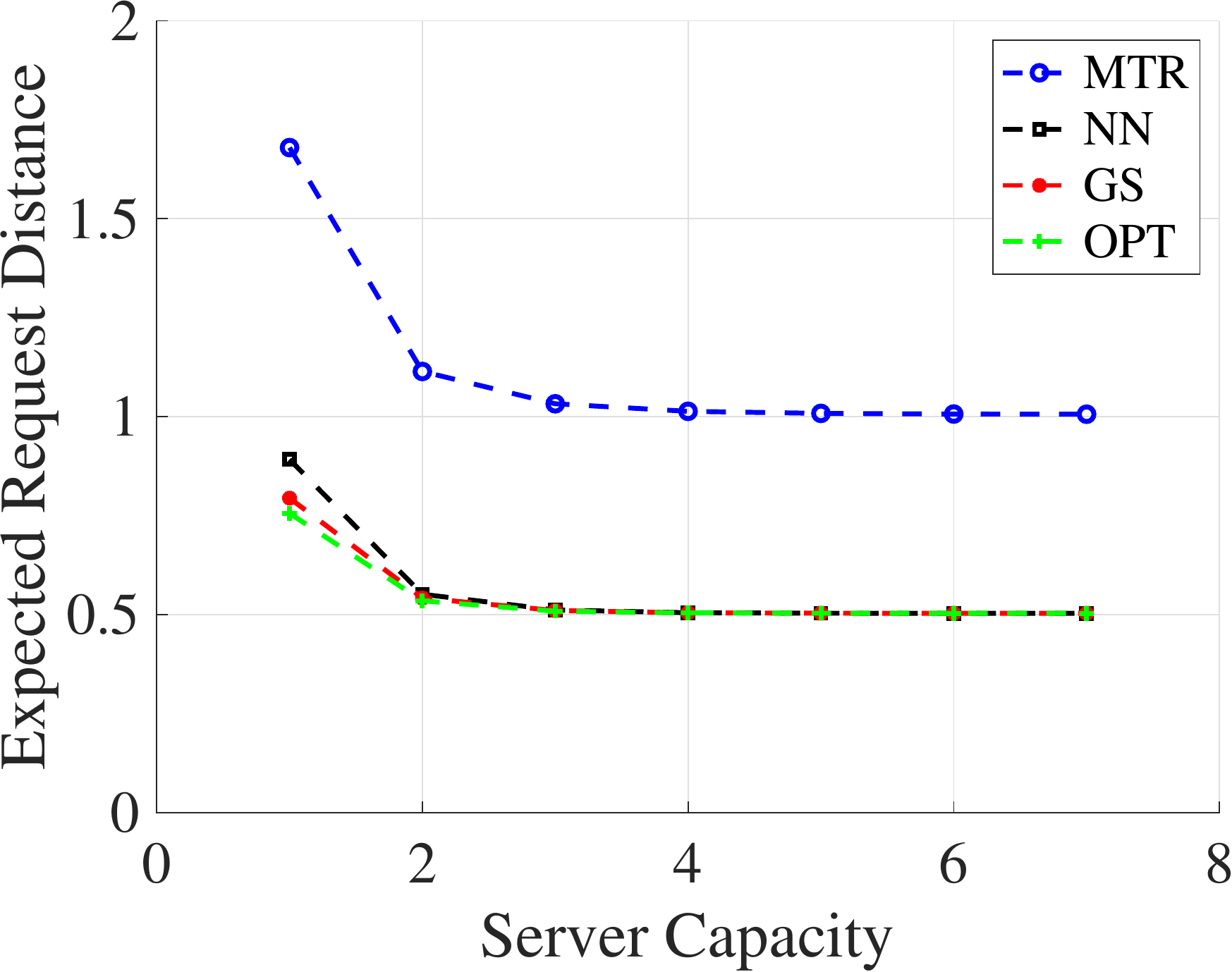}
\subcaption{}
\end{minipage}\hfill
\begin{minipage}{0.25\textwidth}
\includegraphics[width=0.9\textwidth]{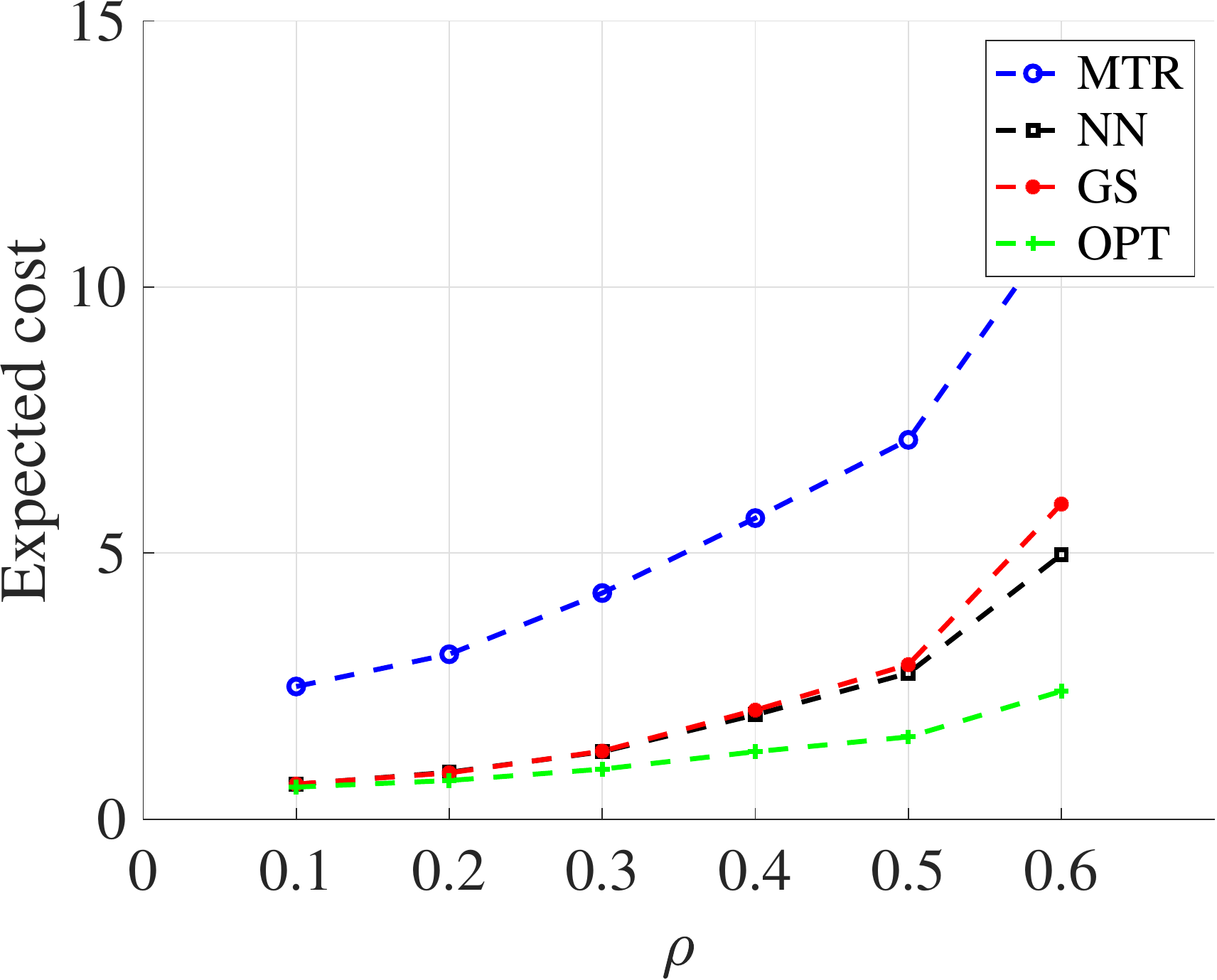}
\subcaption{}
\end{minipage}\hfill
\begin{minipage}{0.25\textwidth}
\includegraphics[width=0.9\textwidth]{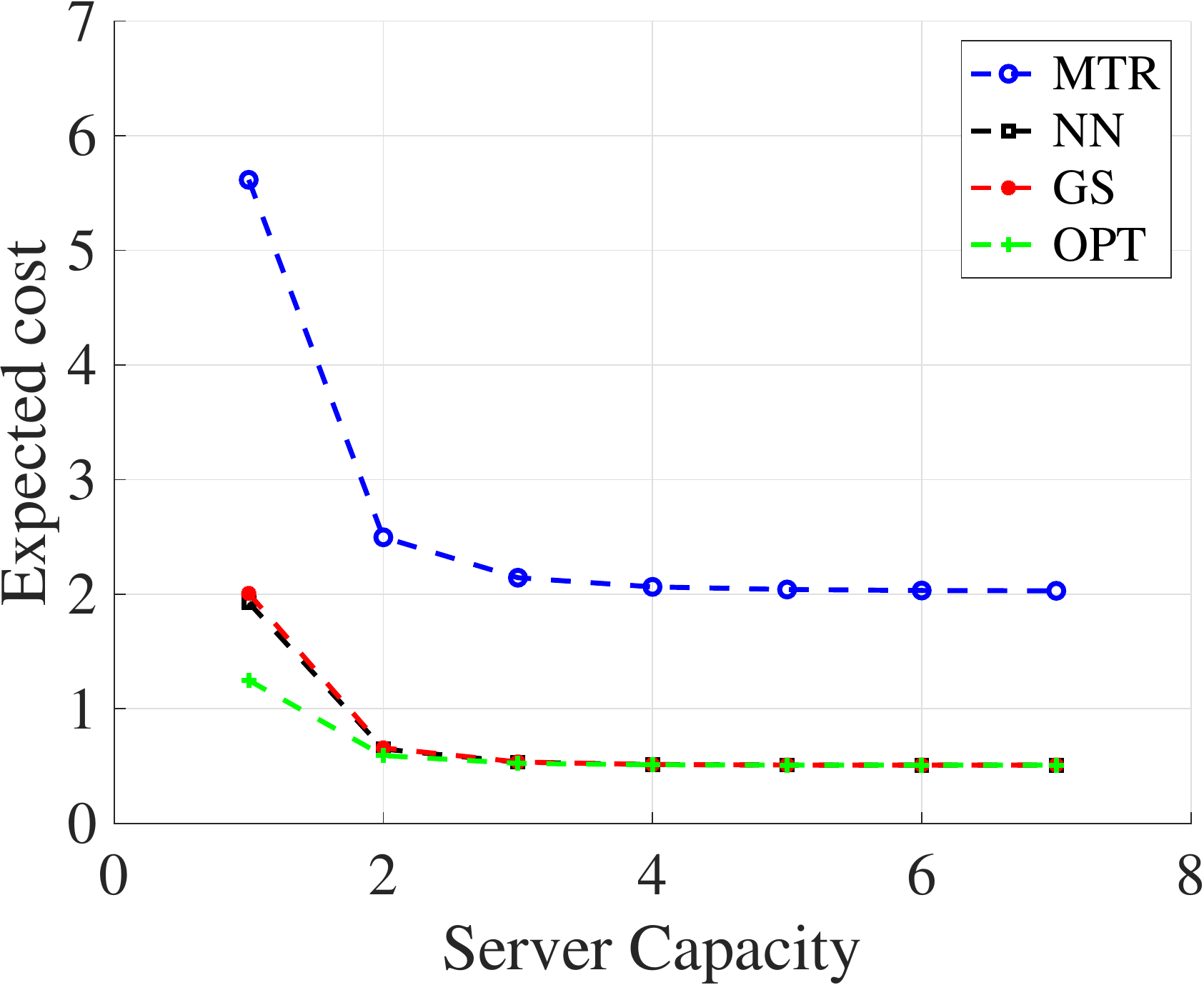}
\subcaption{}
\end{minipage}
\caption{Comparison of different allocation policies: (a) $\pmb{\rho}$ vs $\pmb{\mathbb{E}[D]}$ with $\pmb{c=1}$, (b) $\pmb{c}$ vs. $\pmb{\mathbb{E}[D]}$ with $\pmb{\rho = 0.4}$, \np{(c) $\pmb{\rho}$ vs $\pmb{\overline{T}}$ with $\pmb{\beta = 2, t_0 = 1, c=1}$ and (d) $\pmb{c}$ vs. $\pmb{\overline{T}}$ with $\pmb{\beta = 2, t_0 = 1, \rho=0.4}$.}}
\label{fig:heuristics}
\vspace{-0.16in}
\end{figure}

We consider the case in which both users and servers are distributed according to Poisson processes. From Figure \ref{fig:heuristics} (a), we observe that due to its directional nature MTR has a larger expected request distance compared to other policies while GS provides near optimal performance. At low loads i.e. when $\rho \ll 1$, the Nearest Neighbor policy policy performs similar to the optimal policy. But as $\rho \to 1,$ the NN policy perform worse.

In Figure \ref{fig:heuristics} (b), we compare the performance of allocation policies across different server capacities. The expected request distance decreases with increase in server capacities across all policies. NN,     GS and the optimal policy converge to the same value as $c$ gets higher. 

\np{We now consider the expected communication cost as the performance metric. We use a cost model described in Section \ref{sub:cost_model} with the parameter $\beta  =2$ and $t_0 = 1.$ From Figure \ref{fig:heuristics} (c), we observe that while at low loads i.e. when $\rho \ll 1$, GS and NN perform similar to the optimal policy, as $\rho$ increases both GS and NN perform worse. Note that, NN has a higher expected request distance than GS at high load as shown in Figure \ref{fig:heuristics} (a). However, the performance is reversed with $\beta = 2$, i.e. NN has a lower expected cost than GS at high load  as shown in Figure \ref{fig:heuristics} (c). This depicts the effect of $\beta$ on the performance of various allocation policies. In Figure \ref{fig:heuristics} (d), we observe that NN, GS and the optimal policy converge to the same value as $c$ gets higher.}

We observe similar trends in the case of deterministic inter-server distance distributions. However, under equal densities, all the policies produce smaller  expected request distance as compared to their Poisson counterpart.  This advocates for placing equidistant servers in a bidirectional system with Poisson distributed requesters to minimize expected request distance.

%% file: 10-2d-to-1d.tex
We now consider the case where requesters and servers are located on the two-dimensional plane, $\mathbb{R}^2$. The problem of minimizing the expected request distance can be solved by first constructing a complete $R\times S$ bipartite graph with edge weights $w_{r,s} = \|r,s\|_2, r\in R, s\in S$; followed by executing the Hungarian matching algorithm whose time complexity is $O(n^3)$, where $n=|R|+|S|$\footnote{While the specific case where the weights are Euclidean distances can be solved by Agrawal's algorithm in $O(n^{2+\epsilon})$ time, for a more general weight function the more expensive Hugarian algorithm is needed.}. In this section, we present a heuristic algorithm, which leverages the optimal dynamic programming scheme for one-dimensional inputs to solve the two-dimensional problem, has $O(n^2)$ time complexity, and empirically yields request distances within a constant factor of the optimal solution.

The key insight is to \emph{embed} the points denoting $R,S \subset \mathbb{R^2}$ into new locations in $\mathbb{R}$ such that the distances between a requester $r\in R$ and its $K$ nearest neighbors (servers) $s \in Neighbors(r)$ are approximately preserved. We observe that while distance-preserving or even low-distortion embeddings into a very low dimensional space like $\mathbb{R}$ typically do not exist, embeddings that preserve distances to $K$ nearest neighbors of the other node type (for not too large $K$) may be plausible. This is useful because preserving the nearest servers from $\mathbb{R^2}$ to $\mathbb{R}$ provides a reasonable opportunity for the Dynamic Programming algorithm outlined in Section \ref{sec:opt} to find good matchings.

We achieve the aforementioned embedding by adapting a non-linear dimensionality reduction method such as Locally Linear Embedding (LLE)~\cite{Roweis00}, which consists of the steps outlined below.

\begin{figure}
\centering
\frame{\includegraphics[width=0.5\textwidth]{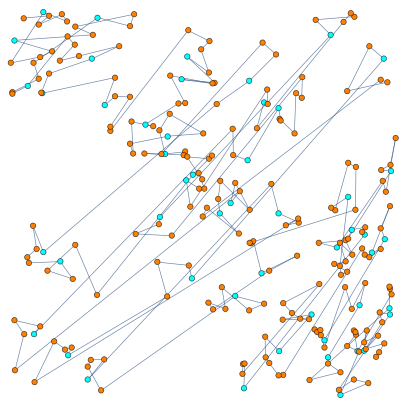}}
\caption{\label{fig:zigzag}Approximate nearest-K-distance preserving $\mathbb{R}^2 \to \mathbb{R}$ embedding ($K=25\%$)}
\end{figure}

\paragraph{Estimation of nearest neighbor weights.}
For each requester $r_i$, select $K$ nearest servers $s_{i1}, s_{i2}, \ldots, s_{iK}$\footnote{Note that in general, the nearest servers need not be the ones with the smallest Euclidean distance from $r_i$; they could be the ones with low \emph{costs} to $r_i$. However in this section, we equate the costs with the Euclidean distance.}. Estimate a set of weights $w_{i1}, w_{i2}, \ldots, w_{iK}$ such that the point $r_i$ can be reconstructed from $s_{i1}, s_{i2}, \ldots, s_{iK}: r_i = \sum_{j=1}^K w_{ij} s_{ij}$. Similarly for each server $s_i$, select $K'$ nearest requesters $r_{i1}, r_{i2}, \ldots, r_{iK'}$ and estimated weights such that point $s_i$ can be reconstructed from the nearest neighbor requester locations: $s_i = \sum_{j=1}^K w_{ij} r_{ij}$.

This can be achieved by minimizing the reconstruction error for each node $i \in R\times S$. Suppose $K$ is fixed for both requesters and servers. $W$ is an $n\times n$ matrix of weights where $n=|R|+|S|$ and the $i$-th row $W_i$, which corresponds to node $i$, has $K$ non-zero elements. The structure of $W$ is as follows:
\[
W = 
\begin{pmatrix}
0 & W_{RS} \\
W_{SR} & 0  
\end{pmatrix},
\]

\noindent where $W_{RS}$ and $W_{SR}$ are rectangular matrices with dimensions $|R|\times |S|$ and $|S|\times |R|$, respectively. If the two-dimensional coordinates of node $i$ are represented by vector $\mathbf{x_i}$, the reconstruction error can be defined by:
\begin{equation}
\epsilon(W_i) = \|\mathbf{x_i} - \sum_{j=1}^K W_{ij} \mathbf{x_j} \|^2.
\end{equation}

It was shown in \cite{Roweis00} that $\epsilon(W_i)$ is minimized when $W_i = (G_i + \lambda I)^{-1} \mathbf{1}$, where $G_i(j,k) = (\mathbf{x_j} - \mathbf{x_i}).(\mathbf{x_k} - \mathbf{x_i})$, $\mathbf{1}$ is the vector of all ones, and $\lambda$ is chosen such that the elements of $W_i$ add up to 1.

\paragraph{Computing optimal embedding in $\mathbb{R}$.}
LLE suggests that the relationships between the $n$ points in the higher dimensional space ($\mathbf{R}^2$ in our case) captured by the matrix $W$ should be approximately preserved in the lower dimensional space ($\mathbf{R}$ in our case). Then the optimal embedding $\mathbf{y} = \{y_1, y_2, \ldots, y_n \}, y_i \in \mathbb{R}$ can be found by solving the following quadratic optimization problem:
\begin{equation}
\begin{array}{l l}
\min & \sum_{i=1}^n (y_i - \sum_{j} W_{ij} y_j)^2 = \mathbf{y}^T (I-W)^T (I-W) \mathbf{y},\\
\textrm{subject to:} & \mathbf{y}^T \mathbf{y} = 1
\end{array}
\end{equation}

\noindent $(I-W)^T (I-W)$ is a positive semi-definite sparse matrix (since $K << n$) and the optimal solution to this ``eigenvalue" problem is given by the eigenvector corresponding to the smallest non-zero eigenvalue of $(I-W)^T (I-W)$~\cite{Roweis00}. Since we do not need to compute all the eigenvectors, the second smallest eigenvalue of a matrix can be computed efficiently without performing a matrix diagonalization using the Arnoldi algorithm in running time $O(n^2)$.


\paragraph{Using the $\mathbb{R}$-embedding for matching.}
After generating the embedding $\mathbf{y}$, we applied our Dynamic Programming Algorithm to compute the best resource allocation scheme. However, naive application of the algorithm led to high expected request distances. The reason behind this is illustrated in Figure \ref{fig:zigzag}, which visualizes an embedding computed for a given set of requesters and servers from $\mathbb{R}^2$ to $\mathbb{R}$; the zigzag lines denote the linear order imposed by the embedding. It is easy to see that it is quite possible that a pair of (requester (blue), server (orange)) nodes which are far away in $\mathbb{R}^2$ may be pretty close to each other in $\mathbb{R}$. Since our LLE-based scheme only tries to preserve close-by neighbors and does not explicitly attempt to repel nodes that are farther away in $\mathbb{R}^2$, such pairs of nodes could end up being embedded close to each other in $\mathbb{R}$. To circumvent this problem, we propose a heuristic scheme to adjust the embedding $\mathbf{y}$ such that whenever for a pair of nodes $\{i,j\}$ we have $\|\mathbf{x_i}-\mathbf{x_j}\| > \Delta$ but $\|y_k - y_{k+1}\| < \epsilon$, where $\mathbf{x_i}$ is mapped to $y_k$ and $\mathbf{x_j}$ is mapped to $y_{k+1}$, and $\Delta, \epsilon, \delta$ are configurable constants, we increase the distance between $y_k$ and $y_{k+1}$ by adding a large cumulative constant $c_{k+1} = c_k + \delta$ to $y_{k+1}$. This adjustment of $\mathbf{y}$ sequentially \emph{spreads out} the points in $\mathbb{R}$ toward the right and the Dynamic Programing Algorithm is then able to find good requester-server matchings.

\begin{figure}
\centering
\frame{\includegraphics[width=0.4\textwidth]{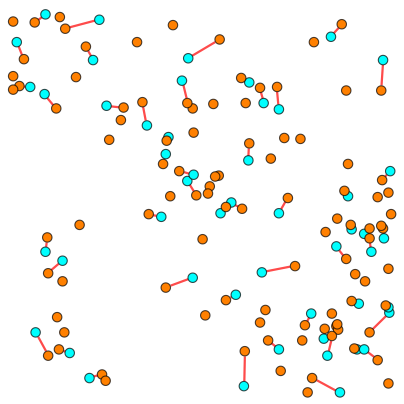}}
\frame{\includegraphics[width=0.4\textwidth]{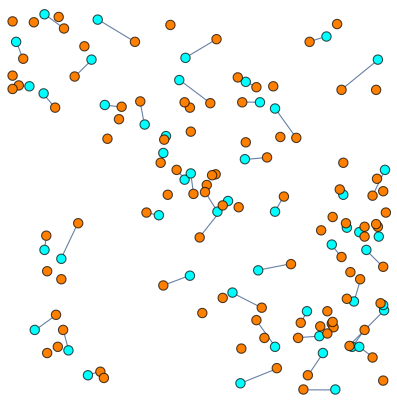}}
\caption{\label{fig:2dto1d}Matching for a clustered distribution: $|R|=50$ requesters are spread uniformly at random in a square $[0,1]\times[0,1]$ and each of the $|S|=100$ servers is located uniformly at random in a box of size $0.1$ around a randomly selected resource. (a) Optimum weighted bipartite matching in 2D ($\mathbb{E}[D]=0.0435$); (b) Approximate matching after 1D embedding ($\mathbb{E}[D]=0.06$) ($K=25\%)$}
\end{figure}

\begin{table*}[]
\center
\begin{tabular}{r  r  l  l}\hline
$|R|$ & $|S|$ & $\mathbb{E}[D]$ &  $\mathbb{E}[D]$ \\
& & Optimum & $\mathbb{R}^2 \to \mathbb{R}$ \\
\hline
$200$ & $400$ & $0.023$ & $0.05$\\
$200$ & $500$ & $0.0205$ & $0.041$\\
$200$ & $600$ & $0.018$ & $0.037$\\
$200$ & $700$ & $0.0166$ & $0.0336$\\
$200$ & $800$ & $0.0158$ & $0.0326$\\
$200$ & $900$ & $0.0151$ & $0.02989$\\
$200$ & $1000$ & $0.014$ & $0.02985$\\
\hline
\end{tabular}
\caption{\label{tab:2dto1d}Matching in larger networks: requesters are spread uniformly at random in a square $[0,1]\times[0,1]$ and each server is located uniformly at random in a box of size $0.1$ around a randomly selected resource.}
\end{table*}

Figure \ref{fig:2dto1d} shows a comparison between an optimum matching and an approximate matching constructed by the embedding methods proposed in this section. Table \ref{tab:2dto1d} shows results for the case when $|S|$ is varied for a fixed $|R|$. We can observe that the $\mathbb{R}^2 \to \mathbb{R}$ embedding approach yields a solution which is empirically close to $2 \times OPT$. Given that this procedure has a lower time complexity $O(n^2)$\footnote{Both the embedding process and the Dynamic Programming algorithm have time complexity $O(n^2)$.} than the usual $O(n^3)$ for Hungarian algorithm, it could be practically useful for large resource allocation problems. 

%% file: 11a-conclusion.tex

We introduced a queuing theoretic model for analyzing the behavior of unidirectional policies to allocate tasks to servers on the real line. We showed the equivalence of UGS and MTR  w.r.t the expected request distance and presented results associated with the case when either requesters or servers were Poisson distributed. In this context, we analyzed a new queueing theoretic model: ESABQ, not previously studied in queueing literature. We also proposed a dynamic programming based algorithm to obtain an optimal allocation policy in a bi-directional system. We performed sensitivity analysis for unidirectional system and compared the  performance of various greedy allocation strategies along with the unidirectional policies to that of optimal policy. \np{We proposed a heuristic based approximate solution to the optimal assignment problem for the two-dimensional scenario.} Going further, we aim to extend our analysis for unidirectional policies to a two-dimensional geographic region.

%% file: 11b-ack.tex
This research was sponsored by the U.S. Army Research Laboratory and the U.K. Defence Science and Technology Laboratory under Agreement Number W911NF-16-3-0001 and by the NSF under grant NSF CNS-1617437. The views and conclusions contained in this document are those of the authors and should not be interpreted as representing the official policies, either expressed or implied, of the U.S. Army Research Laboratory, the U.S. Government, the U.K. Defence Science and Technology Laboratory. This document does not contain technology or technical data controlled under either the U.S. International Traffic in Arms Regulations or the U.S. Export Administration Regulations.

%% file: 11-appendix.tex

\subsection{Derivation of $F_Z$ for various inter-server distance distrbutions}\label{app-Ge-distribs}
\subsubsection{$\pmb{F_X(x)\sim \text{Exponential}(\mu)}$}
In this case, both $X$ and $Y$ are exponentially distributed. Thus the difference distribution is given by
\begin{align}
D_{XY}(x) = 1 - \frac{\lambda}{\lambda+\mu}e^{-\mu x}, \text{when} \;x\ge0\label{eq:diff_exp}
\end{align}

\noindent Combining \eqref{eq:first_busy} and \eqref{eq:diff_exp}, we get  
\begin{align}
F_Z(x) &= \frac{1 - \frac{\lambda}{\lambda+\mu}e^{-\mu x} - 1 + \frac{\lambda}{\lambda+\mu}}{\frac{\lambda}{\lambda+\mu}} = 1 - e^{-\mu x}.
\end{align}

\noindent Thus we obtain $F_X(x) = F_Z(x)\sim \text{Exponential}(\mu).$ 

\subsubsection{$\pmb{F_X(x)\sim \text{Uniform}(0,b)}$} \label{sub:uniform}
The c.d.f. for uniform distribution is 
\begin{equation}\label{eq:utility}
    F_X(x)=
\begin{cases}
     \frac{x}{b},& 0 \le x \le b;\\
     1,& x > b,
\end{cases}
\end{equation}
where $b$ is the uniform parameter. Thus we have 
\begin{align}
D_{XY}(x) = \int_0^{\infty} F_X(x+y)\lambda e^{-\lambda y} dy = \left[\int_0^{b-x} \frac{x+y}{b} \lambda e^{-\lambda y} dy\right] + \left[\int_{b-x}^{\infty} 1 \;\lambda e^{-\lambda y} dy\right] =\frac{\lambda x - e^{-\lambda(b-x)}+e^{-\lambda b}}{b\lambda+e^{-\lambda b}-1}\label{eq:unf1}
\end{align}

\noindent Taking $k_\lambda = 1/(b\lambda+e^{-\lambda b}-1)$ and using Equation \eqref{eq:first_busy} we have 
\begin{align}
F_Z(x) = k_\lambda\left[\lambda x + e^{-\lambda b}(1-e^{\lambda x})\right] \quad \hbox{and} \quad f_Z(x) = \lambda k_\lambda\left[1-e^{-\lambda b}e^{\lambda x})\right].
\end{align}

\noindent Taking $\alpha_Z = \int_{0}^{b}xf_Z(x)dx$ and $\sigma_Z^2 = [\int_{0}^{b}x^2f_Z(x)dx] - \alpha_Z^2$ we have 
\begin{align}
\alpha_Z = \frac{b^2\lambda}{2}k_\lambda - \frac{1}{\lambda},\quad \sigma^2_Z = \frac{b^3\lambda}{3}k_\lambda - \frac{k_\lambda}{\lambda}\left[b(b\lambda-2)+\frac{2}{\lambda}(1-e^{-\lambda b})\right]- \alpha_Z^2, \quad\alpha_X = b/2, \quad \sigma^2_X = b^2/12.
\end{align}

\subsubsection{$\pmb{F_X(x)\sim \text{Deterministic}(d_0)}$} \label{sub:deterministic}
Another interesting scenario is when servers are equally spaced at a distance $d_0$ from each other i.e. when $F_X(x)\sim \text{Deterministic}(d_0).$ The c.d.f. for deterministic distribution is 
\begin{equation}\label{eq:utility}
    F_X(x)=
\begin{cases}
     0,& 0 \le x < d_0;\\
     1,& x \ge  d_0,
\end{cases}
\end{equation}
where $d_0$ is the deterministic  parameter. A similar analysis as that of uniform distribution yields

{\footnotesize
\begin{align}
F_Z(x) &= c_\lambda\left[e^{-\lambda (d_0-x)}-e^{\lambda d_0}\right]; \;f_Z(x) = \lambda c_\lambda\left[e^{-\lambda (d_0-x)}\right],
\end{align}
}
\noindent where $c_{\lambda} = 1/(1-e^{-\lambda d_0}).$ Thus we have 
\begin{align}
\alpha_Z = c_\lambda\frac{d_0\lambda+e^{-\lambda d_0}-1}{\lambda},\quad\sigma^2_Z = \frac{c_\lambda}{\lambda}\left[d_0(d_0\lambda-2)+\frac{2}{\lambda}(1-e^{-\lambda d_0})\right]- \alpha_Z^2,\quad\alpha_X = d_0, \quad \sigma^2_X = 0.
\end{align}

\subsection{ESABQ under PRGS}
\subsubsection{Chapman-Kolmorogov equations}\label{sub:ck}
Let us write the Chapman-Kolmorogov equations for the Markov chain $\{(L(t), R(t), I(t)),\, t\geq 0\}$ defined in Section \ref{sub:ql_prgs}.

For $n\geq 2$ and $x>0$ we get
{\footnotesize
\begin{align}
\frac{\partial}{\partial t}p_t(n,x;1) &=\frac{\partial}{\partial x}p_t(n,x;1)-\lambda p_t(n,x;1)- \frac{\partial}{\partial x} p_t(n,0;1)+\lambda p_t(n-1,x;1)\nonumber\\
\frac{\partial}{\partial t} p_t(n,x;2) &=\frac{\partial}{\partial x} p_t(n,x;2)-\lambda p_t(n,x;2) - \frac{\partial}{\partial x} p_t(n,0;2) +\lambda p_t(n-1,x;2)+ F_X(x) \frac{\partial}{\partial x} p_t(n+c,0;1) + F_X(x) \frac{\partial}{\partial x} p_t(n+c,0;2)\nonumber.
\end{align}
}
Letting $t\to\infty$ yields
\begin{align}
0&=\frac{\partial}{\partial x} p(n,x;1)-\lambda p(n,x;1)- \frac{\partial}{\partial x} p(n,0;1)+\lambda p(n-1,x;1) \label{eq:1} \\
0&=\frac{\partial}{\partial x} p(n,x;2)-\lambda p(n,x;2) - \frac{\partial}{\partial x}p(n,0;2)+ \lambda p(n-1,x;2) + F_X(x) \frac{\partial}{\partial x} p(n+c,0;1)+ F_X(x) \frac{\partial}{\partial x} p(n+c,0;2).\label{eq:2}
\end{align}
For $n=1$, $x>0$
\begin{align}
\frac{\partial}{\partial t} p_t(1,x;1)&=\frac{\partial}{\partial x} p_t(1,x;1)-\lambda p_t(1,x;1)-\frac{\partial}{\partial x}p_t(1,0;1)+\lambda p_t(0)F_Z(x)\nonumber\\
\frac{\partial}{\partial t} p_t(1,x;2)&=\frac{\partial}{\partial x} p_t(1,x;2)-\lambda p_t(1,x;2)-\frac{\partial}{\partial x}p_t(1,0;2)+F_X(x) \frac{\partial}{\partial x} p(1+c,0;1)+ F_X(x)\frac{\partial}{\partial x} p_t(1+c,0;2).\nonumber
\end{align}
Letting $t\to\infty$ yields
\begin{align}
0&=\frac{\partial}{\partial x} p(1,x;1)-\lambda p(1,x;1)-\frac{\partial}{\partial x}p(1,0;1)+\lambda p(0)F_Z(x) \label{eq:3}\\
0&=\frac{\partial}{\partial x} p(1,x;2)-\lambda p(1,x;2)-\frac{\partial}{\partial x}p(1,0;2) + F_X(x) \left(\frac{\partial}{\partial x} p(1+c,0;1)+ \frac{\partial}{\partial x} p(1+c,0;2)\right), x>0.
\label{eq:4}
\end{align}
We can collect the results in (\ref{eq:1})-(\ref{eq:4}) as follows: for $n\geq 1$, $x>0$,
\begin{align}
0&=\frac{\partial}{\partial x} p(n,x;1)-\lambda p(n,x;1)- \frac{\partial}{\partial x} p(n,0;1)+ \lambda p(n-1,x;1){\bf 1}(n\geq 2)+\lambda p(0)F_Z(x){\bf 1}(n=1) \label{eq:20}\\
0&=\frac{\partial}{\partial x} p(n,x;2)-\lambda p(n,x;2) - \frac{\partial}{\partial x}p(n,0;2) + \lambda p(n-1,x;2){\bf 1}(n\geq 2)+ F_X(x)\left( \frac{\partial}{\partial x} p(n+c,0;1)+ \frac{\partial}{\partial x} p(n+c,0;2)\right). \label{eq:20-bis}
\end{align}

Define $g(n,x)=p(n,x;1)+p(n,x;2)$ for $n\geq 1$, $x>0$. Summing (\ref{eq:20}) and (\ref{eq:20-bis}) gives
\begin{align}
0=&\frac{\partial}{\partial x} g(n,x)-\lambda g(n,x)- \frac{\partial}{\partial x} g(n,0)+
\lambda g(n-1,x){\bf 1}(n\geq 2)+\lambda p(0)F_Z(x){\bf 1}(n=1)+ F_X(x) \frac{\partial}{\partial x} g(n+c,0), \nonumber\\& \forall n\geq 1, x>0.
\label{eq:211}
\end{align}
\subsubsection{Multiplicity of roots of $z^c - F^*_X(\lambda(1-z))$}\label{sub:rem1}
Assume that $F_X(x)=1-e^{-\mu x}$ (regular batch service times are exponentially distributed). Then,
\[
z^c - F^*_X(\lambda(1-z))=\frac{-\rho z^{c+1} +(1+\rho)z^c-1}{1+\rho(1-z)}.
\]
$z^c - F^*_X(\lambda(1-z))=0$ for $|z| \leq 1$ iff $Q(z):=-\rho z^{c+1} +(1+\rho)z^c-1=0$. 
The derivative of $Q(z)$ is $Q'(z)=z^{c-1}((1+\rho)c -\rho(c+1)z)$. It vanishes at $z=0$ and at $z=\frac{(1+\rho)c}{\rho(c+1)}>1$ under the stability condition $\rho<c$. Since $z=0$ is not a zero of $Q(z)$, we conclude that all zeros of $z^c - F^*_X(\lambda(1-z))$  in $\{|z|\leq 1\}$ have multiplicity one.

More generally, it is shown in \cite{bailey54} that all zeros of $z^c-F^*_X(\lambda(1-z))$ in $\{|z|\leq 1\}$ have multiplicity one if $F_X$ is a $\chi^2$-distribution with an even number $2p$ of degrees of freedom, i.e. $dF_X(x)= \frac{a^p}{\Gamma(p)}x^{p-1}e^{-ax} dx$ so that $1/\mu= p/a$.

\subsubsection{Roots of $A(z)$}\label{app-mg1-rouche}
Define $A(z)=F^*_X(\theta(z))$. If $A(z)$ has a radius of convergence larger than one (i.e. $A(z)$ is analytic for $|z|\leq \nu$ with $\nu>1$) and $A'(1)<c\in \{1,2,\ldots\}$  a direct application of Rouch\'e's theorem shows that $z^c - A(z)$ has $c$ zeros in the unit disk $\{|z|\leq 1\}$(see e.g. \cite{Adan05}). If the radius of convergence of $A(z)$ is one, $A(z)$ is differentiable at $z=1$, $A'(1)<c$,  and $z^c-A(z)$ has period $p$, then $z^c -A(z)$ has exactly $p\leq s$ zeros on the unit circle and $s-p$ zeros inside the unit disk $\{|z|<1\}$ \cite[Theorem 3.2]{Adan05}. Assume that the stability condition $\frac{d}{dz}A(z)|_{z=1}=\rho<c$ holds. $A(z)$ has a radius of convergence larger than one when $F_X$ is the exponential/Erlang/Gamma/ etc probability distributions.

\subsubsection{Special Cases}\label{app:verify}
One easily checks that (\ref{value:Nz}) gives the classical Pollaczek-Khinchin formula for the M/G/1 queue when $c=1$ and $F_Z=F_X$.

Let now $c=1$ in (\ref{value:Nz}) with $F_Z$ and $F_X$ arbitrary. Then,
\[
N(z)=\frac{a_1}{\lambda} \left(\frac{F^*_X(\lambda(1-z))-zF^*_Z(\lambda(1-z))}{F^*_X(\lambda(1-z))-z}\right)
\]
gives the $z$-transform of the stationary number of customers  in a M/G/1 queue with an exceptional first customer in a busy period.  The constant $a_1/\lambda$ is obtained from the identity $N(1)=1$ by application of L'Hopital's rule, which  gives\footnote{Note that we retrieve this result by letting $c=1$ in  (\ref{eq:c}).} $a_1/\lambda=(1-\rho)/(1-\rho+\rho_Z)$. This gives
\[
N(z)=\frac{1-\rho}{1-\rho+\rho_Z} \left(\frac{F^*_X(\lambda(1-z))-zF^*_Z(\lambda(1-z))}{F^*_X(\lambda(1-z))-z}\right).
\]
The above is a known result  \cite{Welch64}. 

If $F^*_Z=F^*_X:=F^*$, then 
\[
N(z)
=\frac{ \sum_{k=1}^c a_k\left[(z^c-z^k)z^c+((1-z^c)z-(1-z^k)) F^*(\theta(z))\right]}{\theta(z) (z^c-F^*(\theta(z))}.
\]


\subsection{Results for Section \ref{sec:het}}\label{app-het}
\begin{figure}[]
\centering
\includegraphics[width=0.6\columnwidth]{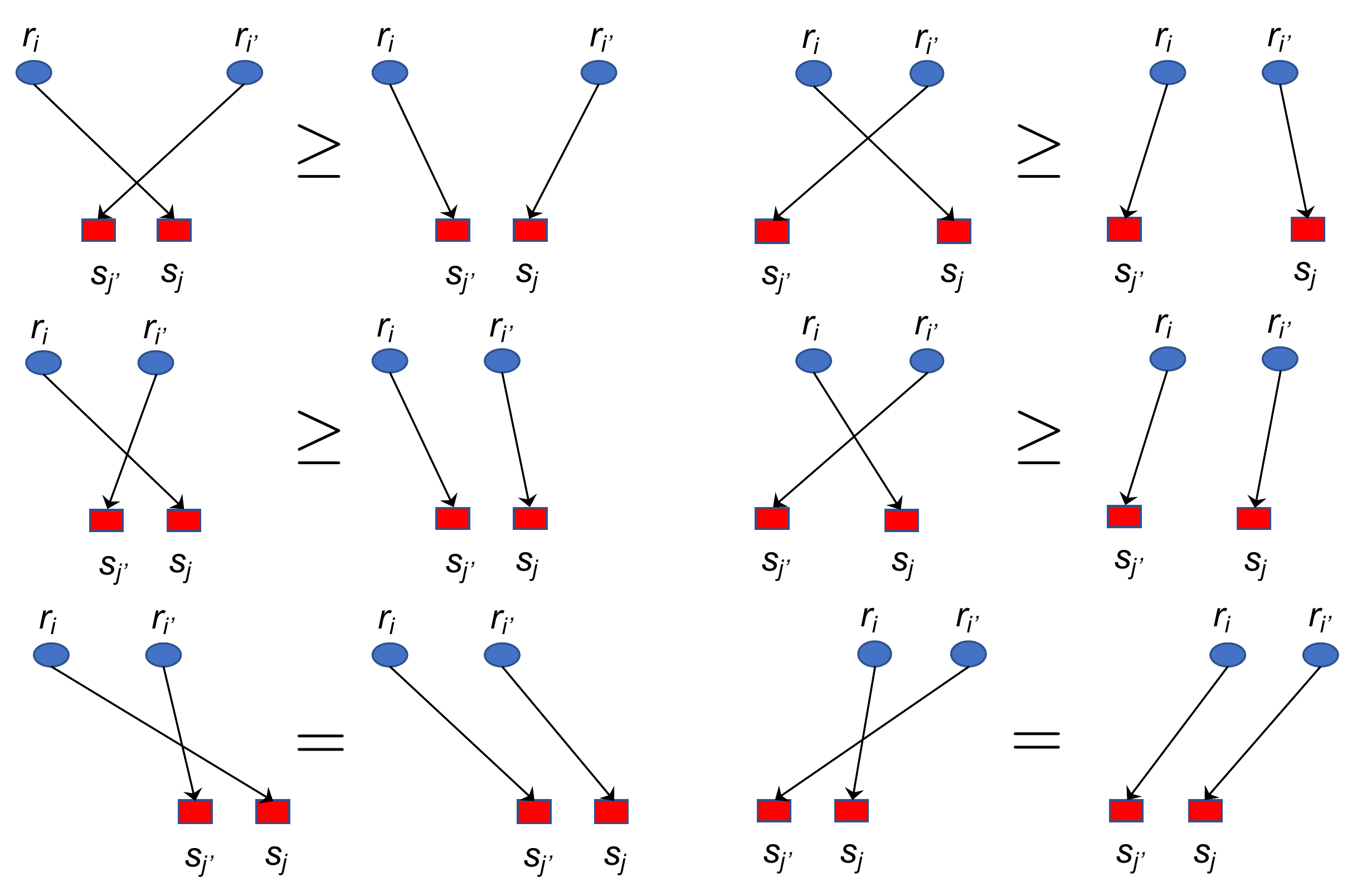}
\caption{\label{fig:nocross}Uncrossing an assignment either reduces request distance or keeps it unchanged.}
\end{figure}
\subsubsection{Derivation of $v_1(z)$ and $v_2(z)$}\label{app-het-rouche}
$v_1(z)$ in \eqref{eq:lhs_var_c} can further be simplified to
{\footnotesize
\begin{align}
v_1(z) &= \sum\limits_{l=0}^\infty z^l \sum\limits_{m=0}^{l}\pi_m\sum\limits_{i=0}^{c}k_{i+l-m}p_i = \sum\limits_{m=0}^\infty \pi_m\sum\limits_{l\ge m} z^l\sum\limits_{i=0}^{c}k_{i+l-m}p_i= \sum\limits_{m=0}^\infty \pi_mz^m\sum\limits_{l\ge m} z^{l-m}\sum\limits_{i=0}^{c}k_{i+l-m}p_i = \sum\limits_{m=0}^\infty \pi_mz^m\sum\limits_{j=0}^{\infty} z^{j}\sum\limits_{i=0}^{c}k_{i+j}p_i \nonumber\\
&= N(z)\sum\limits_{i=0}^{c}p_iz^{-i}\sum\limits_{j=0}^{\infty} z^{i+j}k_{i+j}= N(z)\sum\limits_{i=0}^{c}p_iz^{-i}\bigg[K(z)-\sum\limits_{j=0}^{i}k_jz^j+k_iz^i\bigg]= N(z)\bigg\{\sum\limits_{i=0}^{c}p_iz^{-i}\bigg[K(z)-\sum\limits_{j=0}^{i}k_jz^j\bigg]+\sum\limits_{i=0}^{c}k_iz^i\bigg\}.
\end{align} 
}

Similarly $v_2(z)$ in \eqref{eq:rhs_var_c} can further be simplified to

{\footnotesize
\begin{align}
&v_2(z) = \sum\limits_{l=0}^\infty z^l \sum\limits_{m=l+1}^{c+l}\pi_m\sum\limits_{i=m-l}^{c}k_{i+l-m}p_i = \bigg[\sum\limits_{l=0}^\infty z^l \sum\limits_{m=l}^{c+l}\pi_m\sum\limits_{i=m-l}^{c}k_{i+l-m}p_i\bigg] - N(z)\sum\limits_{i=0}^{c}k_iz^i = \bigg[\sum\limits_{m=0}^{c}z^{-m}\sum\limits_{i=m}^{c}k_{i-m}p_i\sum\limits_{l=0}^\infty z^{m+l} \pi_{m+l}\bigg] - N(z)\sum\limits_{i=0}^{c}k_iz^i\nonumber\\
&= \bigg[\sum\limits_{m=0}^{c}z^{-m}\sum\limits_{i=m}^{c}k_{i-m}p_i\bigg\{N(z)-\sum\limits_{j=0}^{m-1}\pi_jz^j\bigg\}\bigg] - N(z)\sum\limits_{i=0}^{c}k_iz^i.
\end{align} 
}

\subsubsection{Roots of $A(z)$}\label{app-het-rouche}
Denote $A(z) = K(z)\sum_{i=0}^{c}p_{c-i}z^{i}.$ Clearly, $A(z)$ is also a probability generating function ({\it pgf}) for the non-negative random variable $V+\tilde{\mathcal C}$ where $\tilde{\mathcal C}$ is a random variable on $\{0,\ldots,c-1\}$ with distribution $\texttt{Pr}(\tilde{\mathcal C}=j) = p_{c-j}, \forall j\in\{0,1,\ldots,c-1\}.$ Also we have
\begin{align}
A^{\prime}(1) &= K^{\prime}(1)+\sum_{i=0}^{c}p_{c-i}i = \rho+\sum_{i=1}^{c}p_{i}(c-i) = \rho+\sum_{i=1}^{c}p_{i}c-\sum_{i=1}^{c}ip_i = \rho+c-\overline{\mathcal C}\nonumber
\end{align}
From our stability condition we know that $\rho < \overline{\mathcal C}.$ Thus $A^{\prime}(1) < c.$ Since $A(z)$ is a pgf and $A^{\prime}(1) < c$, by applying the arguments from \cite[Theorem 3.2]{Adan05} we conclude that the denominator of equation \eqref{eq:het_l} has $c-1$ zeros inside and one on the unit circle, $|z|= 1.$

\subsection{Proof of Lemma \ref{lem:nocross}}\label{app-lemma}
\begin{proof}
It can be observed that if such a 4-tuple $(i,j,i',j')$ exists, the cost can be reduced by assigning $i$ to $j'$ and $i'$ to $j$, hence we arrive at a contradiction. To show this, consider the six possible cases of relative ordering between $r_i, r_{i'}, s_j, s_{j'}$ which obey $r_i < r_{i'}$ and $s_j > s_{j'}$. We give a pictorial proof in Figure \ref{fig:nocross}\footnote{For ease of exposition, the requesters and servers are shown to be located along two separate horizontal lines, although they are located on the same real-line.}. It is easy to see that in each of the cases, the request distance of the \emph{uncrossed} assignment is either smaller or remains unchanged.
\end{proof}